\documentclass[lettersize,journal]{IEEEtran}
\usepackage{amsmath,amsfonts,amssymb,mathrsfs}
\usepackage{algorithmic}
\usepackage{algorithm}
\usepackage{array}
\usepackage[caption=false,font=normalsize,labelfont=sf,textfont=sf]{subfig}
\usepackage{textcomp}
\usepackage{stfloats}
\usepackage{url}
\usepackage{verbatim}
\usepackage{graphicx}
\usepackage{cite}
\usepackage{epstopdf}
\usepackage{amsthm}
\usepackage{color}
\theoremstyle{definition}

\newtheorem{theorem}{Theorem}[section]
\newtheorem{lemma}{Lemma}[section]
\newtheorem{remark}{Remark}[section]

\newtheorem{assumption}{Assumption}[section]
\newtheorem{definition}{Definition}[section]

\begin{document}
	
	\title{Fully distributed consensus control for stochastic multi-agent systems under undirected and directed topologies}
	
	\author{Xuping Hou, Xiaofeng Zong, \IEEEmembership{Member, IEEE}, Yong He, \IEEEmembership{Senior Member, IEEE}
			\thanks{The research was supported by the National Natural Science Foundation of China under Grants 62473347 and 62522319, and by the Hubei Provincial Outstanding Young and Middle-Aged Science and Technology Innovation Team, T2024032.}
			\thanks{X. Hou, X. Zong, and Y. He are with the School of Artificial Intelligence and Automation, China University of Geosciences, Wuhan 430074, China, and Hubei Key Laboratory of Advanced Control and Intelligent Automation for Complex Systems, Wuhan 430074, China.(e-mail:  houxuping@cug.edu.cn, zongxf@cug.edu.cn, heyong08@cug.edu.cn)}
		}
		
		
		
		\maketitle
		
		\begin{abstract}
			This work aims to address the design of fully distributed control protocols for stochastic consensus, and, for the first time, establishes the existence and uniqueness of solutions for the path-dependent and highly nonlinear closed-loop systems under both undirected and directed topologies, bridging a critical gap in the literature.
			For the case of directed graphs, a unified fully distributed control protocol is designed for the first time to guarantee mean square and almost sure consensus of stochastic multi-agent systems under directed graphs.
			Moreover, an enhanced fully distributed protocol with additional tunable parameters designed for undirected graphs is proposed, which guarantees stochastic consensus while achieving superior convergence speed. Additionally, our work provides explicit exponential estimates for the corresponding convergence rates of stochastic consensus, elucidating the relationship between the exponential convergence rate and the system parameters. Simulations validate the theoretical results.
			
		\end{abstract}
		

		\section{Introduction}
		In recent years, researchers have discovered that individual agents with limited communication and perception capabilities can effectively collaborate with each other in teams to accomplish more complex and diverse tasks that exceed their individual capabilities, and to some extent, save resources and costs \cite{zhu2013flocking, ge2025distributed}. Consequently, the distributed collaborative control of multi-agent systems (MASs) has found widespread applications in various fields, including geological and marine exploration, agricultural automation, as well as military and aerospace \cite{fax2004information, olfati2006flocking, ren2008distributed, zhang2015satellite, dong2015time}. Among these applications, the consensus problem \cite{olfati2004consensus} stands out as a typical and fundamental issue in distributed collaborative control.
		
		Nevertheless, the design of control protocols discussed in the above works relies on Laplacian matrix information, which is a global information about the structure of communication topologies. This reliance imposes a significant limitation by hindering fully distributed control protocols and consequently undermining some advantages inherent to distributed control approaches. To overcome this constraint, Li et al. proposed for the first time a fully distributed protocol for consensus and tracking control of MASs, marking a significant shift toward local-information-based strategies \cite{li2013distributed_tac}. Since then, researchers have made substantial progress in more complex situations through numerous exemplary contributions in the area of fully distributed control of deterministic MASs \cite{shafiee2018multi, sahafi2025fully}.

	A key development in this field is the widespread adoption of adaptive gain design, which enables dynamic control adjustment using only local information. Recent studies demonstrate this approach across diverse network topologies and system models. For instance, in \cite{mei2016distributed}, the consensus problem of second-order MASs with heterogeneous unknown time-varying inertias and control gains was studied under directed graphs through adaptive $\sigma$-modification schemes.
			In \cite{xu2022fully}, the problem of fully distributed consensus for MASs with general linear dynamics under undirected graphs was solved via adaptive dynamic event-triggered communication strategies.
			Further extending to signed networks, a time-varying parametric Lyapunov-based adaptive protocol was proposed in \cite{zhou2024fully} for prescribed-time bipartite synchronization in cooperative-antagonistic networks.
			Meanwhile, distributed adaptive mechanisms were employed in \cite{zuo2023fully} and \cite{wang2021fully} to achieve practical fixed-time consensus for single-integrator systems under undirected graphs and leaderless/leader-following consensus for second-order systems under directed graphs, respectively. 
			Moreover, a topology-agnostic, data-driven event-triggered adaptive learning algorithm was introduced in \cite{ma2025event} for cooperative control under model uncertainty and communication constraints.
		Collectively, these works demonstrate that the adaptive mechanisms can effectively operate across diverse network structures, ranging from undirected and directed graphs to signed digraphs and topology-agnostic scenarios.
			However, it should be noted that the aforementioned protocols primarily focus on deterministic MASs, leaving the challenge of stochastic disturbances largely unaddressed in current fully distributed control frameworks.
		
		In modern engineering systems, interference is common and inevitable, and is caused by many factors such as complex communication environments, abrupt changes in working conditions, aging or damage of equipment, friction, and so on. Ignoring interference often results in an inability to meet high-precision control requirements in most practical applications. Therefore, considering the stochastic noise interference in the dynamics or communication topology of MASs in complex environments, researchers typically employ $\mathrm{It}\hat{\mathrm{o}} $ stochastic differential equations to model the dynamic behavior of each agent\cite{mao2007stochastic,nourian2012mean,shaikhet2013lyapunov}. In recent years, there has been a wealth of outstanding research addressing stochastic MASs, and many studies in this area can be found in \cite{ma2017consensus,kaviarasan2020non,jia2024time}.
		
		However, as far as we know, research on stochastic fully distributed consensus appears to be quite limited in the existing literature.
		In \cite{gu2020adaptive}, the synchronization problem of a stochastic coupled nonlinear dynamic complex network under adaptive proportional-integral control was investigated.
		In \cite{li2022fully}, a fully distributed tracking control protocol was proposed for the stochastic nonlinear stochastic MASs with Markov switching topology by designing intermittent adaptive gains.
		Additionally, in \cite{wen2024fully}, the fully distributed bipartite time-varying formation tracking control problem for heterogeneous linear MASs with stochastic disturbances was investigated under a signed Markovian switching topology.
		However, the aforementioned literature primarily focuses on synchronization or tracking control problems in the mean square(m.s.) sense and the interaction topologies among the nodes or followers are assumed to be undirected.
		The conventional analytical techniques for undirected graphs are proven inadequate for the case of directed graphs due to their reliance on symmetric properties--a fundamental feature absent in directed networks. This methodological incompatibility becomes particularly pronounced in stochastic cases, where two critical challenges arise: i) the existence of solutions for stochastic systems under fully distributed protocols requires rigorous verification, and ii) the almost sure(a.s.) and m.s. consensus in stochastic environment need theoretical examination. These inherent complexity of conducting consensus analysis and protocol synthesis for directed graphs under stochastic conditions leaves this a challenging problem.
		At present, to the best of our knowledge, for stochastic MASs, no existing work has addressed the design of fully distributed control protocols for directed graphs while ensuring stochastic a.s. consensus, which is one of the motivations for our work.
		
		Notably, the implementation of fully distributed control protocols leads to path-dependent and highly nonlinear closed-loop dynamics \cite{li2013distributed_auto}, since the feedback involves integral functionals of the state history, which poses significant challenges for consensus analysis in stochastic MASs.
		For this case, it is necessary to prove the existence and uniqueness of the solution of closed-loop stochastic system in both undirected and directed graphs. However, although there have been some studies on fully distributed control of stochastic MASs under undirected graphs, no prior work has explored this fundamental issue, which is another motivation for this work.
		
		
		Based on the preceding analysis, the contributions of this work are summarized as follows:
		
		i) A unified fully distributed protocol is proposed to solve the m.s. and a.s. consensus problems for stochastic MASs under directed topology. To the best of our knowledge, this represents the first systematic study addressing fully distributed control for stochastic MASs under directed graphs, filling a notable gap in existing literature.
		
		ii) For stochastic MASs under undirected graphs, an improved fully distributed protocol with rigorous convergence guarantees is developed. In contrast to prior works on fully distributed stochastic consensus, our analysis provides explicit exponential convergence rate estimates for both m.s. and a.s. consensus, establishing clear relationships between convergence performance and system parameters.
		
		iii) By employing stochastic analysis tools including the functional $\mathrm{It\hat{o}} $ formula and stopping-time techniques, we rigorously establish the existence and uniqueness of solutions for a class of highly nonlinear, path-dependent stochastic systems.
		
		This paper is organized as follows. The problem formulation is provided in Section \ref{sec2}. Section \ref{sec-3-1} presents the existence and uniqueness of the solution to the path-dependent stochastic systems.
		The main results on the fully distributed protocol design for stochastic MASs are presented in Section \ref{sec-3-2}. In Section \ref{sec4}, two simulations are provided to validate the theoretical findings. Section \ref{sec5} summarizes the paper.
		
		\textbf{Notations:} $P^T$ represents the transposition of the $P$ matrix. $\mathbb{R}^{n\times m}$ is the set of all $n\times m$ real matrices. $P>0$ for $P\in \mathbb{R}^{n\times n}$ means that $P$ is a positive definite matrix. $\lambda_{\min}(P)$ and $\lambda_{\max}(P)$ represent the minimum and maximum eigenvalues of $P$. $I_n$ denotes the $n$-dimensional identity matrix. Let $(\Omega, {\mathcal{F}}, \mathbb{P})$ denote a complete probability space with a filtration $\{{\mathcal{F}}_t\}_{t\geq 0}$ satisfying the usual conditions. For $p, q \in R$, $p\wedge  q$ and $p\vee q$ represent $\min \left \{  p, q\right \} $ and $\max \left \{  p, q\right \} $. $\sigma_{\max}(A)$ denotes the maximum singular value of a matrix $A \in \mathbb{R}^{m \times n}$. $\mathbf{1}_N$ denotes the $N$-dimensional column vector with all ones.
		
		\section{Problem Formulation}\label{sec2}
		The information interaction among different agents can be modeled as an undirected graph $\bar{\mathcal{G}}=\{ \mathcal{V} ,\bar{\mathcal{E}} ,\bar{A }\}$ or a directed graph $\tilde{\mathcal{G}}=\{ \mathcal{V},\tilde{\mathcal{E}} ,\tilde{A }\}$.
		Among them, $\mathcal{V}=\{n_1,n_i,...,n_N\}$ represents the node set with $i$ being the $i$th agent.
		$\bar{A }(or \ \tilde{A })=[a_{ij}]_{N\times N}$ represents the adjacency matrix with $a_{ij}>0$ if $(n_i,n_j) \in \mathcal{E}(or \ \tilde{\mathcal{E}})$, otherwise $a_{ij}=0$.
		The set of agent $i$'s neighbors is represented as $N_{i}$, that is, for $j\in N_{i}$, $a_{ij}>0$. The Laplacian matrix $\mathcal{L}= [\mathcal{L}_{ij}] \in \mathbb{R}^{N\times N}$ of $\mathcal{G}$(or \ $\tilde{\mathcal{G}}$) is denoted as $\mathcal{L}_{ii}=\sum_{j\in N_{i}} a_{ij}$ and $\mathcal{L}_{ij}=- a_{ij}, i \ne j$.
		For an undirected graph $\bar{\mathcal{G}}$, $\bar{\mathcal{E}}=\{(n_i,n_j)|n_i,n_j\in \mathcal{V}\}$ represents the edge set and $(n_i,n_j) \in \mathcal{E}$ implies $(n_j,n_i) \in \mathcal{E}$ for any $n_i,n_j\in \mathcal{V}$. In addition, for a directed graph $\tilde{\mathcal{G}}$, $(n_i,n_j) \in \tilde{\mathcal{E}} $ represents the direction from $n_i$ to $n_j$.
		
		Consider the following stochastic MASs with $N$ nodes where the dynamic of each agent is modeled by
		\begin{equation}\label{2-1}
			\begin{split}
				dx_i(t)=[Ax_i(t)+Bu_i(t)]dt+Cx_i(t)dw(t),
			\end{split}
		\end{equation}
		where $x_i(t) \in \mathbb{R}^n$ and $u_i(t)\in \mathbb{R}^m$ are the state and control input of the $i$th agent, respectively. $i=1,2,...,N$, $A \in \mathbb{R}^{n\times n}$, $B  \in \mathbb{R}^{n\times m}$ and $C  \in \mathbb{R}^{n\times n}$ are constant matrices, and $w(t)$ is a standard Brownian motion defined on the complete probability space $(\Omega, {\mathcal{F}}, \mathbb{P})$.
		Let $x (t)=[x _1^T(t),...,x _{N}^T(t)]^{T}$ and $u(t)=[u_1^T(t),...,u_{N}^T(t)]^{T}$.
		
		The unified fully distributed protocol is given as
		\begin{eqnarray}\label{Section4-1}
			\begin{split}
				&u_i(t)=c_i(t) \Sigma_i(t) \mathcal{K} \xi_i(t),\\
				&\dot{c}_i(t)=e^{\gamma t}\xi_i^T(t)\Gamma \xi_i(t),
			\end{split}
		\end{eqnarray}
		where $\xi_i(t)=\sum_{j\in N_i} a_{ij}(x_i(t)-x _j(t))$, $c_i(t)\in \mathbb{R}$ denotes the time-varying adaptive gain, $c_i(0)> 0$, and $\Sigma_i(t)\in \mathbb{R}$ represents the auxiliary time-varying gain, which is designed as
		$$ \Sigma_i(t)= k_1\Big(k_2+\frac{\sigma_i(t)}{c_i(t)}\Big)^{\mu},$$ with  $\sigma_i(t)=\xi_i^T(t)P\xi_i^T(t)$, $\mu\ge 1, k_1, k_2>0 $. The feedback gain matrices $\mathcal{\mathcal{K} } \in \mathbb{R}^{m\times n}$, $\Gamma \in \mathbb{R}^{n\times n}$, $P \in \mathbb{R}^{n\times n}$ and constant $\gamma$ will be designed later in Section \ref{sec-3-2}.
		
		The goal is to design the control input $u(t)$ such that the consensus of the stochastic MASs \eqref{2-1} can be solved in m.s. or a.s. sense. The corresponding  definitions are as follows.
		\begin{definition} The m.s. (or a.s.) consensus of the stochastic MASs \eqref{2-1} can be solved by
			the protocol $u(t)$  if for any initial value $x(0)\in\mathbb{R}^{nN}$ and all distinct $i,j \in \mathcal{V}$, $\lim_{t \to \infty }\mathbb{E}\|x_i(t)-x _j(t)\|^2=0$ (or $\lim_{t \to \infty }\|x_i(t)-x _j(t)\|=0$, a.s.).
		\end{definition}
		
		Denote $x  (t)=[x _1^T(t),...,x _N^T(t)]^T$. Substituting the protocol \eqref{Section4-1} into \eqref{2-1} yields the following closed-loop system
		\begin{eqnarray}\label{Section4-2-linear}
			\begin{split}
				dx  (t)=&(I_N \otimes A  + (\tilde{C}(t)\tilde{\Sigma}(t)) \mathcal{L} \otimes B  \mathcal{K} )x  (t)dt\\&+(I_N \otimes C)x  (t)dw(t),
			\end{split}
		\end{eqnarray}
		where $\tilde{C}(t)=\mathrm{diag}\{c_1(t),...,c_N(t)\}$, and $\tilde{\Sigma}(t)=\mathrm{diag}\{\Sigma_1(t),...,\Sigma_N(t)\}$.
		Let $\vartheta_i(t)=x_i(t)-\frac{1}{N}\sum_{j=1}^{N} x _j(t)$.
		Denote $\vartheta(t)=[(I_N-\frac{1}{N}\mathbf{1}_N\mathbf{1}_N^T)\otimes I_n]x  (t)$, where $\vartheta(t)=[\vartheta_1^T(t),...,\vartheta_N^T(t)]^T$. Then, \eqref{Section4-2-linear} can be written as
		\begin{eqnarray}
			\begin{split}\label{Section4-3-linear}
				d\vartheta(t)=&(I_N \otimes A  +  (\tilde{C}(t)\tilde{\Sigma}(t)) \mathcal{L} \otimes B  \mathcal{K} )\vartheta(t)dt\\&+(I_N \otimes C)\vartheta(t)dw(t).
			\end{split}
		\end{eqnarray}
		
		%
		%
		
		The adaptive feedback law proposed in this paper introduces a historical integral term (e.g., $\int_0^t |r(s)x(s)|^p ds$ in \cite{mao1996stochastic}), resulting in a closed-loop system described by a path-dependent stochastic differential equation. Since the drift coefficient depends on the entire historical path of the solution, standard existence and uniqueness theorems under Lipschitz conditions no longer apply.
		Although deterministic counterparts can be treated with well-established methods, the stochastic path-dependent case remains theoretically underdeveloped. Therefore, as pointed in \cite{mao1996stochastic}, it is necessary to establish the existence and uniqueness of solutions for this class of path-dependent stochastic systems.
		
		\section{Existence and uniqueness of the solution to the path-dependent stochastic systems}\label{sec-3-1}
		Consider the following path-dependent stochastic system
		\begin{eqnarray}\label{lem-1}
			\begin{split}
				dy(t)\!=&\bigg[\!\int_{0}^{t} e^{\gamma s} y^T\!(s)Q_1y(s)ds )\bigg]\mathfrak{h}(y(t))dt+\mathfrak{f}(y(t))dt\\
				&+\mathfrak{g}(y(t))dw(t),
			\end{split}
		\end{eqnarray}
		where $y(t) \in \mathbb{R}^n$, $\gamma$ are constants, $Q_1 \in \mathbb{R}^{n\times n}$ is constant matrices with $Q_1 >0$, and $w(t)$ is a standard Brownian motion defined on the complete probability space $(\Omega, {\mathcal{F}}, \mathbb{P})$. The nonlinear term $\mathfrak{f}$, $\mathfrak{g}$, $\mathfrak{h}$ satisfy the following assumption.
		\begin{assumption}[Local Lipschitz Condition]\label{ass-local}
			For each $q>0$ and $\forall y_1, y_2 \in \mathbb{R}^n$ with $\|y_1\| \vee \|y_2\| <q$, there exists a constant $L_q > 0$ such that
			\begin{eqnarray}
				\begin{split}
					&\| \mathfrak{f}(y_1) - \mathfrak{f}(y_2) \| \vee \| \mathfrak{g}(y_1) - \mathfrak{g}(y_2) \| \vee \| \mathfrak{h}(y_1) - \mathfrak{h}(y_2) \|\\
					& \le L_q \| y_1 - y_2 \|.\notag
				\end{split}
			\end{eqnarray}
		\end{assumption}
		In the interest of consistency, we introduce the following functional It\^{o} formula \cite{dupire2019}.
		Define $y_t=\{y(\theta):0\le  \theta \le t\}$.
		Let $\overline{T} > 0$ be a fixed terminal time.
		For each $t \in [0,\overline{T}]$, denote by $\Lambda_t$ the set of cadlag functions from $[0,t]$ to $\mathbb{R}^n$, and define
		$\Lambda \equiv \bigcup_{t \in [0,\overline{T}]} \Lambda_t$.
		Let $y$ be a continuous semimartingale process, with $y(t)$ its value at time $t$ and $y_t \in \Lambda_t$ its path over $[0,t]$.
		Let $U: \Lambda \times [0,\overline{T}] \to \mathbb{R}$ be a smooth functional in the sense of Dupire, i.e., $U$ is $\Lambda$-continuous, twice continuously differentiable in the spatial direction (in the sense of $\mathcal{\nabla}_y$) and once continuously differentiable in the temporal direction (in the sense of $\mathcal{\nabla}_t$), with these derivatives themselves $\Lambda$-continuous.
		Then the functional It\^{o} formula (Theorem 3.1 of \cite{dupire2019}) states that for any $t \in [0,\overline{T}]$:
		\begin{equation}
			\begin{split}
				U(y_t,t) &= U(y_0,0) + \int_0^t \mathcal{\nabla}_y U(y_s,s) \, dy(s) \\
				&\quad + \int_0^t \mathcal{\nabla}_t U(y_s,s) \, ds + \frac{1}{2} \int_0^t \mathcal{\nabla}_{yy} U(y_s,s) \, d\langle y \rangle_s,\notag
			\end{split}
		\end{equation}
		where $\mathcal{\nabla}_y U(y_t,t)$ and $\mathcal{\nabla}_{yy} U(y_t,t)$ are the functional space derivatives, $\mathcal{\nabla}_t U(y_t,t)$ is the functional time derivative, and $d\langle y \rangle_t$ is the quadratic variation of $y$. For precise definitions, see \cite{dupire2019}.
		
		Based on the above functional It\^{o} formula in \cite{dupire2019}, for path-dependent stochastic system \eqref{lem-1} and a functional $U: \Lambda \times [0,\overline{T}] \to \mathbb{R}$,
		we have
		\begin{equation}
			\begin{split}
				U(y_t,t) =& U(y_0,0)+\int_0^t \mathfrak{L} U(y_s,s) ds\\
				&+\int_0^t \mathcal{\nabla}_yU(y_s,s)\mathfrak{g}(y(s))dw(s),\notag
			\end{split}
		\end{equation}
		where $\mathfrak{L} U(y_t,t) $ is defined as
		$\mathfrak{L} U(y_t,t) = \mathcal{\nabla}_tU(y_t,t)+\mathcal{\nabla}_yU(y_t,t)\bar{\mathfrak{f}}(y_t,t)
		+\frac{1}{2}{\rm{trace}}[\mathfrak{g}^T(y(t))\nabla_{yy}U(y_t,t)\mathfrak{g}(y(t))]$ with
		$\bar{\mathfrak{f}}$ depends on the whole path $y_t$, $\bar{\mathfrak{f}}(y_t,t)=\big[\int_{0}^{t} e^{\gamma s} y^T(s)Q_1y(s)ds )\big]\mathfrak{h}(y(t))+\mathfrak{f}(y(t))$.
		
		To establish the existence and uniqueness of the solution to \eqref{lem-1}, the following lemma is introduced.
		\begin{lemma}\label{lemma-1}
			Assume that there exist a functional $U: \Lambda \times [0,\overline{T}] \to \mathbb{R}$ and four constants $\alpha \ge 0$, $\beta \ge 0$, $k\le 2$ and $\lambda_1\ge 0$ such that\\
			(i) $e^{-\gamma t} \big(\!\!\int_{0}^{t} e^{\gamma s} y^T\!(s)Q_1y(s)ds-\beta\big)^2 \vee \lambda_1|y(t)|^2 \le k U(y_t,t)$;\\
			(ii) $\mathfrak{L} U(y_t,t) \!\le \!\!-\alpha U(y_t,t) \!+\phi(t)$,\\
			where $\phi(t)\!=\!\frac{\alpha}{2}e^{-\gamma t}\big(\!\int_{0}^{t} e^{\gamma s}y^T\!(s)Q_1y(s)ds\!-\!\beta\big)^2$, $\gamma$ and $Q_1$ are defined in \eqref{lem-1}. Then the path-dependent stochastic system \eqref{lem-1} exists a unique global solution for any initial value $y(0)\in\mathbb{R}^{n}$.
		\end{lemma}
		\begin{proof}
			By Assumption \ref{ass-local}, the drift and diffusion terms of the path-dependent stochastic system \eqref{lem-1} satisfy the local Lipschitz condition. Therefore, for any initial value $y(0)\in\mathbb{R}^{n}$, there exists a unique maximal local solution $y(t)$ on $t\in[0, T_{\infty})$ to \eqref{lem-1} due to the local Lipschitz continuity of the coefficient
			(see \cite{mao1996stochastic, nguyen2020stability}), where $T_{\infty}$ represents the explosion time.
			Let $h_0>0$. For each integer $h \ge h_0$, define the stopping time $\tau_h=\inf\{t\in[0, T_{\infty}): |y  (t)|\ge h\}$. Denote $\tau_\infty=\lim_{h\to \infty }\tau_h$, where $\tau_\infty \le T_\infty $ a.s. That is, if we can prove $\tau_\infty=\infty$ a.s., then $T_\infty=\infty$ follows, which means the existence and uniqueness of the solution.
			By the functional $\mathrm{It\hat{o} } $ formula and condition (ii), we have for any $h \ge h_0$
			\begin{eqnarray}\label{1-EV-1}
				\begin{split}
					&\mathbb{E} U(y_{\tau_h \wedge t},\tau_h \wedge t) - U(y_0,0) \\
					&\le \mathbb{E}\int_0^{\tau_h \wedge t }\left( -\alpha U(y_s,s)+\phi(s)\right)ds,
				\end{split}
			\end{eqnarray}
			with $\phi(s)=\frac{\alpha}{2}e^{-\gamma s}\big(\int_{0}^{s} e^{\gamma r} y^T(r)Q_1y(r)ds-\beta\big)^2$.
			Note that $\phi(t)=\frac{\alpha}{2}e^{-\gamma t}\big(\int_{0}^{t} e^{\gamma s}y^T(s)Q_1y(s)ds-\beta\big)^2\le \frac{\alpha k}{2} U(y_t,t)$ from condition (i). Therefore, substituting this into \eqref{1-EV-1} yields
			\begin{eqnarray}\label{1-EV-2}
				\begin{split}
					&\mathbb{E} U(y_{\tau_h \wedge t},\tau_h \wedge t) \\
					&\le U(y_0,0) -\alpha(1-\frac{k}{2}) \mathbb{E}\int_0^{\tau_h \wedge t }U(y_s,s)ds\\
					&\le U(y_0,0),
				\end{split}
			\end{eqnarray}
			since $\alpha \ge 0$ and $k\le 2$. Define $\theta_h=  \inf_{|y|\ge h,t\ge 0} {U}(y)$, then
			\begin{eqnarray}\label{1-EV-3}
				\begin{split}
					\theta_h\mathbb{P}(\tau_h\le t)\le U(y_0,0).
				\end{split}
			\end{eqnarray}
			Since $\lim_{|y| \to \infty}\inf_{t\ge 0} {U}(y_t,t) \ge \lim_{|y| \to \infty}\inf_{t\ge 0} \lambda_1|y|^2=\infty$ from condition (i), we have
			\begin{eqnarray}
				\begin{split}
					\theta_h \to \infty  \ as \ h \to \infty.\notag
				\end{split}
			\end{eqnarray}
			Combining with \eqref{1-EV-3}, ones get
			\begin{eqnarray}
				\begin{split}
					\lim_{h \to \infty}\mathbb{P}(\tau_h\le t)=0, \notag
				\end{split}
			\end{eqnarray}
			which implies, $\mathbb{P}(\tau_\infty\le t)=0$, namely $\mathbb{P}(\tau_\infty> t)=1$. Since $t\ge 0$ is arbitrary, $\mathbb{P}(\tau_\infty=\infty)=1$. Therefore, the existence and uniqueness of the solution can be proved.
		\end{proof}
		
		\begin{remark}
			Lemma \ref{lemma-1} established a sufficient condition for proving the existence of a global solution to the path-dependent stochastic system \eqref{lem-1}. Its core mechanism relies on constructing a Lyapunov functional $U(y_t,t)$ satisfying the inequality $\mathfrak{L} U(y_t,t) \leq -\alpha U(y_t,t) + \phi(t)$, where the path-dependent term $\phi(t)$ is instantaneously and pointwise bounded by $U(y_t,t)$ via the coupling condition (i). While this condition is pivotal to the proof, it constitutes the primary restriction, as it requires the path-dependent term $\phi(t) $ to be strictly scaled by the current state energy $U(y_t,t)$ at every moment.
			
		\end{remark}
		Based on this lemma, the subsequent consensus analysis in this paper is conducted under the premise of the existence and uniqueness of the solution, thereby ensuring the rigor of the theoretical analysis.
		
		\section{Fully distributed consensus control for stochastic MASs}\label{sec-3-2}
		In this section, the unified fully distributed protocol \eqref{Section4-1} is specifically designed to solve the consensus problem of stochastic MASs for the cases of undirected and directed graphs, respectively.
		
		As we all know that the algebraic Riccati equation is an important tool for the control design of linear deterministic systems \cite{bittanti1991riccati, 1995Robust}. For linear stochastic systems, we introduce the following stochastic algebraic Riccati equation(SARE) to design the controller parameters
		\begin{equation}\label{2-2}
			\begin{split}
				A^TP+PA-PBB^TP+C^TPC+I_n=0.
			\end{split}
		\end{equation}
		Note that the solvability of this SARE has been well investigated in \cite{rami2000linear}, which demonstrates that the existence and uniqueness of the positive definite solution $P$ is equivalent to that the corresponding stochastic system $dx  (t)=[A x  (t)+B  u(t)]dt+C   x  (t)dw(t)$ is m.s. stabilizable.
		So in this work, SARE \eqref{2-2} is always assumed to have a positive-definite solution $P$.
		
		\subsection{Fully distributed consensus control for stochastic MASs with the directed topology}
		\begin{lemma}[\cite{wang2007new}]\label{lem-decomposition}
			Let $\tilde{\mathcal{G}}$ be a directed graph that contains a spanning tree. Then the node set $\mathcal{V}$ can be partitioned into two subsets $\mathcal{V}_l$ and $\mathcal{V}_f$ satisfying the following properties:
			\begin{enumerate}
				\item[(i)] $\mathcal{V}_l \cup \mathcal{V}_f = \mathcal{V}$ and $\mathcal{V}_l \cap \mathcal{V}_f = \emptyset$;
				\item[(ii)] the subgraph induced by $\mathcal{V}_l$ is strongly connected;
				\item[(iii)] there is no edge from any node in $\mathcal{V}_l$ to any node in $\mathcal{V}_f$.
			\end{enumerate}
		\end{lemma}
		
		By relabeling the nodes such that $\mathcal{V}_l = \{n_1,\dots,n_M\}$ and $\mathcal{V}_f = \{n_{M+1},\dots,n_N\}$ with $1 \le M \le N$, the Laplacian matrix $\mathcal{L} \in \mathbb{R}^{N \times N}$ of $\tilde{\mathcal{G}}$ can be written in the following block form:
		\[
		\mathcal{L} = \begin{bmatrix}
			\mathcal{L}_{11} & 0 \\
			\mathcal{L}_{21} & \mathcal{L}_{22}
		\end{bmatrix},
		\]
		where $\mathcal{L}_{11} \in \mathbb{R}^{M \times M}$ is the Laplacian matrix of the strongly connected subgraph induced by $\mathcal{V}_l$, $\mathcal{L}_{22} \in \mathbb{R}^{(N-M)\times(N-M)}$ is a nonsingular $M$-matrix, and $\mathcal{L}_{21}$ represents the connectivity between $\mathcal{V}_l$ and $\mathcal{V}_f$.
		
		\begin{lemma}[\cite{mei2016distributed}]\label{lem-strongly-connected}
			Let $\mathcal{L}_{11} \in \mathbb{R}^{M \times M}$ with $M\ge 1$ be the Laplacian matrix of the strongly connected subgraph induced by $\mathcal{V}_l$. Then there exists a positive diagonal matrix $R = \operatorname{diag}\{r_1,\dots,r_M\}$ such that
			$\tilde{\mathcal{L}}_{11} = R\mathcal{L}_{11} + \mathcal{L}_{11}^T R$ is symmetric positive semidefinite. Moreover, its eigenvalues can be written as
			$0 = \lambda_1(\tilde{\mathcal{L}}_{11}) < \lambda_2(\tilde{\mathcal{L}}_{11}) \le \cdots \le \lambda_M(\tilde{\mathcal{L}}_{11})$.
			In particular, we can choose $R$ as $r = [r_1,\dots,r_M]^T$ with $\sum_{i=1}^M r_i=1$ and $r^T \mathcal{L}_{11} = 0$.
			
		\end{lemma}
		
		\begin{lemma}[\cite{li2015designing}]\label{lem-nonsingular-M}
			Let $\mathcal{L}_{22} \in \mathbb{R}^{(N-M)\times(N-M)}$ with $M\le N$ be the nonsingular $M$-matrix. Then there exists a positive diagonal matrix $S = \operatorname{diag}\{s_1,\dots,s_{N-M}\}$ such that $\tilde{\mathcal{L}}_{22} = S\mathcal{L}_{22} + \mathcal{L}_{22}^T S$ is symmetric positive definite. Consequently, its eigenvalues satisfy $0 < \lambda_1(\tilde{\mathcal{L}}_{22}) \le \lambda_2(\tilde{\mathcal{L}}_{22}) \le \cdots \le \lambda_{N-M}(\tilde{\mathcal{L}}_{22})$.
			A particular choice of matrix $S$ is $s = [s_1,\dots,s_{N-M}]^T = (\mathcal{L}_{22}^T)^{-1}\mathbf{1}_{N-M}$.
		\end{lemma}
		Let $\xi(t)=[\xi_1^T(t),...,\xi_N^T(t)]^T$ and $\tilde{\sigma}(t)=\mathrm{diag}\{\sigma_1(t),...,$ $\sigma_N(t)\}$. By Lemma \ref{lem-decomposition}, $\xi(t)$ can be partitioned into $\xi_{l}(t)=[\xi_{1}^T(t),\dots,$ $\xi_{M}^T(t)]^T$ and $\xi_f(t)=[\xi_{M+1}^T(t),\dots,$ $\xi_{N}^T(t)]^T$, $\vartheta(t)$ can be partitioned into $\vartheta_{l}(t)=$ $[\vartheta_{1}^T(t),$ $\dots,\vartheta_{M}^T(t)]^T \text{and}$ $\vartheta_f(t)$ $=[\vartheta_{M+1}^T(t),\dots,\vartheta_{N}^T(t)]^T$,  the adaptive gain $\tilde{C}(t)$ can be partitioned into $\tilde{C}_l(t)=\mathrm{diag}\{c_1(t),...,c_M(t)\}$ and $\tilde{C}_f(t)=\mathrm{diag}\{c_{M\!+\!1}(t),...,c_N(t)\!\}, \text{the auxiliary time-varying gain} \ \tilde{\Sigma}(t)$ can be partitioned into $\tilde{\Sigma}_l(t)=\mathrm{diag}\{\Sigma_1(t),...,\Sigma_M(t)\}$ and $\tilde{\Sigma}_f(t)=\mathrm{diag}\{\Sigma_{M+1}(t),...,\Sigma_N(t)\}$, $\tilde{\sigma}(t)$ can be partitioned into $\tilde{\sigma}_l(t)=\mathrm{diag}\{\sigma_1(t),...,\sigma_M(t)\}$ and $\tilde{\sigma}_f(t)=\mathrm{diag}\{\sigma_{M+1}(t),...,\sigma_N(t)\}$, corresponding to the two subsets $\mathcal{V}_l$ and $\mathcal{V}_f$.
			Based on the closed-loop stochastic system \eqref{Section4-3-linear}, we obtain
			$d\vartheta_l(t)=(I_M \otimes A  +  (\tilde{C}_l(t)\tilde{\Sigma}_l(t) \mathcal{L}_{11}) \otimes B  \mathcal{K} )\vartheta_l(t)dt+(I_M \otimes C)\vartheta_l(t)dw(t)$, and $d\vartheta_f(t)=(I_{N-M} \otimes A  +  (\tilde{C}_f(t)\tilde{\Sigma}_f(t)\mathcal{L}_{22}) \otimes B  \mathcal{K} )\vartheta_f(t)dt+  (\tilde{C}_f(t)\tilde{\Sigma}_f(t)\mathcal{L}_{21} \otimes B  \mathcal{K} )\vartheta_l(t)dt+(I_{N-M} \otimes C)\vartheta_f(t)dw(t)$.
		Note that $\xi_l(t)=(\mathcal{L}_{11}\otimes I_n )\vartheta_l(t)$ and $\xi_f(t)=(\mathcal{L}_{21}\otimes I_n )\vartheta_l(t)+(\mathcal{L}_{22}\otimes I_n )\vartheta_f(t)$. Then, we can get
		\begin{eqnarray}\label{xi-l}
			\begin{split}
				d\xi_l(t)=&(I_M \otimes A  +  (\tilde{C}_l(t)\tilde{\Sigma}_l(t) \mathcal{L}_{11}) \otimes B  \mathcal{K} )\xi_l(t)dt\\&+(I_M \otimes C)\xi_l(t)dw(t),
			\end{split}
		\end{eqnarray}
		and
		\begin{eqnarray}\label{xi-f}
			\begin{split}
				d\xi_f(t)=&(I_{N-M} \otimes A  +  (\tilde{C}_f(t)\tilde{\Sigma}_f(t)\mathcal{L}_{22}) \otimes B  \mathcal{K} )\xi_f(t)dt\\&+  (\tilde{C}_l(t)\tilde{\Sigma}_l(t)\mathcal{L}_{21} \otimes B  \mathcal{K} )\xi_l(t)dt\\
				&+(I_{N-M} \otimes C)\xi_f(t)dw(t).
			\end{split}
		\end{eqnarray}
		For the case of $\mu>1$ and $\gamma=0$, we have the following theorem.
		\begin{theorem}\label{theorem-1}
			Suppose the directed graph $\tilde{\mathcal{G}}$ contains a spanning tree. The m.s. and a.s. consensus of \eqref{2-1} can be solved by the protocol \eqref{Section4-1} with $\gamma=0$, $\mathcal{\mathcal{K} } =-B  ^TP$, and $\Gamma=PB  B  ^TP$, where $P$ is the solution to the SARE \eqref{2-2}.
			The auxiliary time-varying gain $\Sigma_i(t)$ satisfies \\
			(I) $\Sigma_i(t)= k_1(k_2+\frac{\sigma_i(t)}{c_i(t)})^{\mu}$, with $\mu>1, k_1, k_2\ge 1$;\\
			(II) $c_i(0)\ge 1$, $\Sigma_i(0)\ge 1$.
		\end{theorem}
		
		\begin{proof}
			Based on the closed-loop stochastic system \eqref{Section4-3-linear}, \eqref{xi-l} and \eqref{xi-f}, we choose the following Lyapunov functional
			\begin{align}
				V(\xi(t), t) =V_l(\xi(t), t)+V_f(\xi(t), t)\notag
			\end{align}
			with $V_l(\xi(t), t)=\eta_1^l\sum_{i=1}^{M} r_i c_i(t)\int_0^{\sigma_i(t)}k_1(k_2+\frac{s}{c_i(t)})^{\mu}ds+  $ $\eta_2^l\sum_{i=1}^{M} r_i (c_i(t)-\psi_i)^2$ and $V_f(\xi(t), t)=\eta_1^f\sum_{i=M+1}^{N} s_{i-M}$ $ c_i(t)\int_0^{\sigma_i(t)}k_1(k_2+\frac{s}{c_i(t)})^{\mu}ds+  \eta_2^f\sum_{i=M+1}^{N}s_{i-M} (c_i(t)-\psi_i)^2$, where $\sigma_i(t)=\xi_i^T(t)P\xi_i^T(t)$, $\eta_1^l,\eta_2^l,\eta_1^f,\eta_2^f\ge 0$, $r_i,s_i>0$, and $\psi_i$ are positive constants to be determined.
			According to the functional $\mathrm{It\hat{o} } $ formula and condition (I), it can be deduced that $d V(\xi(t),t)$ has the following form
			\begin{eqnarray}
				\begin{split}
					dV&(\xi(t),t)\!=\!\mathfrak{L}V(\xi(t),t) dt\!+\! \frac{\partial V(\xi(t),t)}{\partial \xi}\!(I_N\!\otimes\! C   )\xi(t)dw(t),\notag
				\end{split}
			\end{eqnarray}
			where $\mathfrak{L}V(\xi(t),t)$ is defined as
			\begin{eqnarray}\label{Sectionv2-1-linear}
				\begin{split}
					\mathfrak{L} V(\xi(t),t)=\mathfrak{L}V_l(\xi(t),t)+\mathfrak{L}V_f(\xi(t),t)
				\end{split}
			\end{eqnarray}
			with 
				\begin{equation}
				\begin{split}\mathfrak{L} V_l(\xi(t),t)=&\eta_1^l\sum_{i=1}^{M}r_i\dot{c}_i(t)\int_0^{\sigma_i(t)}k_1(k_2+\frac{s}{c_i(t)})^{\mu}ds\\
					&+2\eta_2^l\sum_{i=1}^{M}r_i(c_i(t)-\psi_i)\dot{c}_i(t) \\
					&+\eta_1^l\sum_{i=1}^{M}r_i c_i(t)\Sigma_i(t)\mathfrak{L}\sigma_i(t)\\
					&-\int_0^{\sigma_i(t)} \mu k_1 \frac{s}{c_i^2(t)}\dot{c}_i(t) ( k_2 + \frac{s}{{c}_i(t)})^{\mu-1} ds,\notag
				\end{split}
			\end{equation}
			 and 
			 \begin{equation}
			 	\begin{split}
			 		\mathfrak{L} V_f(\xi(t),t)&=\eta_1^f\sum_{i=M+1}^{N}\!\!\!s_{i-M}\dot{c}_i(t)\!\!\int_0^{\sigma_i(t)}\!\!k_1(k_2+\frac{s}{c_i(t)})^{\mu}ds\\
			 		&+2\eta_2^f\sum_{i=M+1}^{N}s_{i-M}(c_i(t)-\psi_i)\dot{c}_i(t)\\
			 		&+\eta_1^f\sum_{i=M+1}^{N}s_{i-M} c_i(t)\Sigma_i(t)\mathfrak{L}\sigma_i(t)\\
			 		&-\int_0^{\sigma_i(t)} \mu k_1 \frac{s}{c_i^2(t)}\dot{c}_i(t) ( k_2 + \frac{s}{{c}_i(t)})^{\mu-1} ds.\notag
			 	\end{split}
			 \end{equation} Since $\mu, k_1, k_2>0, \sigma_i(t)\ge 0, \ \text{we have}\  \mathfrak{L}V_l(\xi(t),t) \le$ $\sum_{i=1}^{M}r_i\big[\eta_1^l\Sigma_i(t)\sigma_i(t)+2\eta_2^lc_i(t)-2\eta_2^l\psi_i\big]\dot{c}_i(t)+\eta_1^l\sum_{i=1}^{M}r_i$\\ $ c_i(t)\Sigma_i(t)\mathfrak{L}\sigma_i(t)$, and $\mathfrak{L} V_f(\xi(t),t)
			\le \sum_{i=M+1}^{N}s_{i-M}\big[\eta_1^f$ $\Sigma_i(t)\sigma_i(t)+2\eta_2^fc_i(t)-2\eta_2^f\psi_i\big]\dot{c}_i(t)+\eta_1^f\sum_{i=M+1}^{N}s_{i-M} c_i(t)$ $\Sigma_i(t)\mathfrak{L}\sigma_i(t)$.
			Let $\tilde{\psi}=\mathrm{diag}\{\psi_1,...,\psi_N\}$. According to Lemma \ref{lem-decomposition}, $\tilde{\psi}$ can be partitioned into $\tilde{\psi}_l=\mathrm{diag}\{\psi_1,...,\psi_M\}$ and $\tilde{\psi}_f=\mathrm{diag}\{\psi_{M+1},...,\psi_N\}$.
			Let $\mathcal{K} =-B  ^TP$ and $\Gamma=PB  B  ^TP$.
			Note that $\sum_{i=1}^{M}\sigma_i(t)=\xi_l^T(t)(I_{M}\otimes P)\xi_l(t)$ and $\xi_l(t)=(\mathcal{L}_{11} \otimes I_n)\vartheta_l(t)$. By the definition of $\sigma_i(t)$, \eqref{xi-l}, \eqref{xi-f} and using Lemma \ref{lem-decomposition} and \ref{lem-strongly-connected}, we have
			\begin{eqnarray}\label{VVV-1}
				\begin{split}
					&\sum_{i=1}^{M}r_i c_i(t)\Sigma_i(t)\mathfrak{L}\sigma_i(t)\\
					&
					= \xi_l^T\!(t)[R(\tilde{C}_l(t)\tilde{\Sigma}_l(t))\!\otimes\! (A ^T\!P\!+\!PA \!+\!C^T\!PC   )] \xi_l(t)\!\!\!\!\\ &\ \ \ -\!\xi_l^T(t)[(\tilde{C}_l(t)\tilde{\Sigma}_l(t))^2\tilde{\mathcal{L}}_{11}\otimes \Gamma] \xi_l(t),
				\end{split}
			\end{eqnarray}
			and
			\begin{eqnarray}\label{Sectionv2-5-linear}
				\begin{split}
					&\sum_{i=1}^{M}r_i[\eta_1^l\Sigma_i(t)\sigma_i(t)+2\eta_2^l c_i(t)-2\eta_2^l \psi_i]\dot{c}_i(t)\\
					&=\xi_l^T\!(t)[R(\eta_1^l\tilde{\Sigma}_l(t)\tilde{\sigma}_l(t)+2\eta_2^l(\tilde{C}_l(t)\! -\!\tilde{\psi}_l))\!\otimes\! \Gamma] \xi_l(t),
				\end{split}
			\end{eqnarray}
			Note that $\sum_{i=M+1}^{N}\sigma_i(t)=\xi_f^T(t)(I_{N-M}\otimes P)\xi_f(t)$. Then, using Lemma \ref{lem-decomposition}, \ref{lem-nonsingular-M}, and the elementary inequality: $2a^TOb\le \iota a^TOa+\frac{1}{\iota}b^TOb$, for any positive definite matrix $O\in \mathbb{R}^{n\times n} $, $a,b\in \mathbb{R}^{n}$ and $\iota>0 $, we have
			\begin{eqnarray}\label{VVV-1-1}
				\begin{split}
					&\sum_{i=M+1}^{N}s_{i-M} c_i(t)\Sigma_i(t)\mathfrak{L}\sigma_i(t)\\
					&
					= \xi_f^T\!(t)[S(\tilde{C}_f(t)\tilde{\Sigma}_f(t))\otimes\! (A ^T\!P\!+\!PA \!+\!C^T\!PC   )] \xi_f(t)\!\!\!\!\\ &\ \ \ -\!\xi_f^T(t)[(\tilde{C}_f(t)\tilde{\Sigma}_f(t))^2\tilde{\mathcal{L}}_{22}\otimes \Gamma] \xi_f(t)\\ &\ \ \ -\!2\xi_f^T\!(t)[(\tilde{C}_f(t)\tilde{\Sigma}_f(t))(\tilde{C}_l(t)\tilde{\Sigma}_l(t))S\mathcal{L}_{21}\!\otimes \! \Gamma] \xi_l(t)\\
					&
					\le \xi_f^T\!(t)[S(\tilde{C}_f(t)\tilde{\Sigma}_f(t))\otimes\! (A ^T\!P\!+\!PA \!+\!C^T\!PC   )] \xi_f(t)\!\!\!\!\\ &\ \ \ -\xi_f^T(t)[(\tilde{C}_f(t)\tilde{\Sigma}_f(t))^2\tilde{\mathcal{L}}_{22}\otimes \Gamma] \xi_f(t)\\ &\ \ \ +\iota\xi_f^T\!(t)[(\tilde{C}_f(t)\tilde{\Sigma}_f(t))^2\!\otimes \! \Gamma] \xi_f(t)\\ &\ \ \ +\frac{\sigma_{\max}^2(S\mathcal{L}_{21})}{\iota}\xi_l^T\!(t)[(\tilde{C}_l(t)\tilde{\Sigma}_l(t))^2\!\otimes \! \Gamma] \xi_l(t),
				\end{split}
			\end{eqnarray}
			and
			\begin{eqnarray}\label{Sectionv2-5-linear-1}
				\begin{split}
					&\sum_{i=M+1}^{N}s_{i-M}[\eta_1^f\Sigma_i(t)\sigma_i(t)+2\eta_2^f c_i(t)-2\eta_2^f \psi_i]\dot{c}_i(t)\\
					&=\xi_f^T\!(t)[S(\eta_1^f\tilde{\Sigma}_f(t)\tilde{\sigma}_f(t)\!+\!2\eta_2^f(\tilde{C}_f(t)\! -\!\tilde{\psi}_f))\!\otimes\! \Gamma] \xi_f(t),
				\end{split}
			\end{eqnarray}
			By substituting \eqref{VVV-1}-\eqref{Sectionv2-5-linear-1} into \eqref{Sectionv2-1-linear}, we can obtain
			\begin{eqnarray}\label{Sectionv2-5.1-linear}
				\begin{split}
					\mathfrak{L}& V(\xi(t),t)\\
					\le&\eta_1^l\xi_l^T\!(t)[R(\tilde{C}_l(t)\tilde{\Sigma}_l(t))\otimes (A ^T\!P\!+\!PA \!+\!C^T\!PC  )] \xi_l(t)\\
					&\!\!\!\!\!+\!\eta_1^f\xi_f^T\!(t)[S(\tilde{C}_f(t)\tilde{\Sigma}_f(t))\!\otimes\! (A ^T\!P\!+\!PA \!+\!C^T\!PC   )] \xi_f(t)\\ &+\xi_l^T\!(t)\!\Big[\Phi_t^l \otimes\! \Gamma\Big] \!\xi_l(t)+\xi_f^T\!(t)\!\Big[\Phi_t^f \otimes\! \Gamma\Big] \!\xi_f(t),
				\end{split}
			\end{eqnarray}
			with $\Phi_t^l=\eta_1^lR\tilde{\Sigma}_l(t)\tilde{\sigma}_l(t)+2\eta_2^lR(\tilde{C}_l(t) -\tilde{\psi}_l)+\eta_1^f\frac{\sigma_{\max}^2(S\mathcal{L}_{21})}{\iota}(\tilde{C}_l(t)\tilde{\Sigma}_l(t))^2-\eta_1^l\tilde{\mathcal{L}}_{11}(\tilde{C}_l(t)\tilde{\Sigma}_l(t))^2$ and $\Phi_t^f=\eta_1^fS\tilde{\Sigma}_f(t)\tilde{\sigma}_f(t)+2\eta_2^fS(\tilde{C}_f(t) -\tilde{\psi}_f)+ \eta_1^f\iota(\tilde{C}_f(t)\tilde{\Sigma}_f(t))^2-\eta_1^f\tilde{\mathcal{L}}_{22}(\tilde{C}_f(t)\tilde{\Sigma}_f(t))^2$.
			Note that $\Sigma_i(t)= k_1(k_2+\frac{\sigma_i(t)}{c_i(t)})^{\mu}$, for $\mu>1, k_1, k_2\ge 1$ and $c_i(t)\ge 1$, $\Sigma_i(t)\ge 1$ by condition (II). Therefore, we have $\tilde{\sigma}_l(t)\le \tilde{C}_l(t)\tilde{\Sigma}_l(t)^{\frac{1}{\mu}}$ and $\tilde{\sigma}_f(t)\le \tilde{C}_f(t)\tilde{\Sigma}_f(t)^{\frac{1}{\mu}}$. Together with the Young's inequality: $ab \le \epsilon  a^{p} + \frac{(\epsilon p)^{-q/p}}{q}  b^{q}$ with $p, q > 1$ and $\frac{1}{p} + \frac{1}{q} = 1$ for any $\epsilon > 0$ and $a, b \ge 0$, we have
			\begin{eqnarray}
				\begin{split}
					\Phi_t^l \le & \eta_1^l R (\tilde{C}_l(t)\tilde{\Sigma}_l(t))^{1+\frac{1}{\mu}}+2\eta_2^l R(\tilde{C}_l(t)\tilde{\Sigma}_l(t) -\tilde{\psi}_l)\\&-\Big(\eta_1^l\tilde{\mathcal{L}}_{11}-\eta_1^f\frac{\sigma_{\max}^2(S\mathcal{L}_{21})}{\iota}I_M\Big)(\tilde{C}_l(t)\tilde{\Sigma}_l(t))^2\\
					\le &\eta_1^l \Big(\epsilon_1  [(\tilde{C}_l(t)\tilde{\Sigma}_l(t))^{1+\frac{1}{\mu}}]^{p_1} + \frac{(\epsilon_1 p_1)^{-q_1/p_1}}{q_1}  R^{q_1}\Big)\\
					&+ 2\eta_2^l\Big(\epsilon_2  (\tilde{C}_l(t)\tilde{\Sigma}_l(t))^{p_2} + \frac{(\epsilon_2 p_2)^{-q_2/p_2}}{q_2}  R^{q_2}\Big)\\
					&-\! 2\eta_2^l R \tilde{\psi}_l \!-\!\!\Big(\!\eta_1^l\lambda_2(\tilde{\mathcal{L}}_{11})\!-\!\eta_1^f\frac{\sigma_{\max}^2(S\mathcal{L}_{21})}{\iota}\!\Big)\!(\tilde{C}_l(t)\tilde{\Sigma}_l(t))^2\!. \notag
				\end{split}
			\end{eqnarray}
			Note that we can choose $p_1$ and $p_2$ to satisfy $(1+\frac{1}{\mu})p_1\le 2$ and $1 < p_2\le 2$ since $\mu>1$ and $p_1,p_2>1$. Combining these conditions yields
			\begin{eqnarray}
				\begin{split}
					\Phi_t^l
					\!\le &\Big[\!\eta_1^l\epsilon_1\!+\!2\eta_2^l\epsilon_2\!-\!\!\Big(\!\eta_1^l \!\lambda_2(\tilde{\mathcal{L}}_{11})\!-\!\eta_1^f\frac{\sigma_{\max}^2\!(\!S\mathcal{L}_{21}\!)}{\iota}\!\Big)\!\Big]\!(\!\tilde{C}_l(t)\tilde{\Sigma}_l(t)\!)^2\\
					&- \!2\eta_2^l R \tilde{\psi}_l\!+\!\eta_1^l \frac{(\epsilon_1 p_1)^{-q_1/p_1}}{q_1}  R^{q_1}\!+ \!2\eta_2^l \frac{(\epsilon_2 p_2)^{-q_2/p_2}}{q_2}  R^{q_2}. \notag
				\end{split}
			\end{eqnarray}
			Let $\psi_i \!\ge\! \frac{\eta_1^l(\epsilon_1 p_1)^{-q_1/p_1}}{2\eta_2^l q_1} \{r_i\}^{q_1-1}\!+  \frac{(\epsilon_2 p_2)^{-q_2/p_2}}{q_2}  \{r_i\}^{q_2-1}+\frac{\check{\psi_i}}{2\eta_2^l r_i} $ for $i=1,\dots, M$, where $\check{\psi_i}$, $i=1,\dots, M$ will be designed later. Therefore, we obtain
			\begin{eqnarray}\label{Sectionv2-5.2-linear}
				\begin{split}
					\Phi_t^l
					&\le -\eta_3^l(\tilde{C}_l(t)\tilde{\Sigma}_l(t))^2-\check{\psi}_l.
				\end{split}
			\end{eqnarray}
			with $\eta_3^l=\eta_1^l\lambda_2(\tilde{\mathcal{L}}_{11})-\eta_1^f\frac{\sigma_{\max}^2(S\mathcal{L}_{21})}{\iota}-\eta_1^l\epsilon_1-2\eta_2^l\epsilon_2$ and $\check{\psi}_l=\mathrm{diag}\{\check{\psi}_1,...,\check{\psi}_M\}$.
			Let $\psi_i \!\ge\! \frac{\eta_1(\epsilon_1 p_1)^{-q_1/p_1}}{2\eta_2 q_1} \{s_i\}^{q_1-1}\!+  \frac{(\epsilon_2 p_2)^{-q_2/p_2}}{q_2}  \{s_i\}^{q_2-1}+\frac{\check{\psi_i}}{2\eta_2 s_i} $ for $i=M+1,\dots, N$, where $\check{\psi_i}$, $i=M+1,\dots, N$ will be designed later. Similarly, we have
			\begin{eqnarray}\label{Sectionv2-5.2-linear-1}
				\begin{split}
					\Phi_t^f
					&\le -\eta_3^f(\tilde{C}_f(t)\tilde{\Sigma}_f(t))^2-\check{\psi}_f.
				\end{split}
			\end{eqnarray}
			with $\eta_3^f=\eta_1^f(\lambda_1(\tilde{\mathcal{L}}_{22})-\iota)-\eta_1^f\epsilon_1-2\eta_2^f\epsilon_2$ and $\check{\psi}_f=\mathrm{diag}\{\check{\psi}_{M+1},...,\check{\psi}_N\}$. Choose $\eta_1^f<\frac{\eta_1^l\lambda_1(\tilde{\mathcal{L}}_{22})\lambda_2(\tilde{\mathcal{L}}_{11})}{\sigma_{\max}^2(S\mathcal{L}_{21})}$, $\iota <\lambda_1(\tilde{\mathcal{L}}_{22})$, and sufficiently small $\epsilon_1,\epsilon_2$ to ensure $\eta_3^l>0$ and $\eta_3^f>0$.
			Substituting the above inequalities \eqref{Sectionv2-5.2-linear} and \eqref{Sectionv2-5.2-linear-1} into \eqref{Sectionv2-5.1-linear} and using the elementary inequality: $2a^TOb\le \iota a^TOa+\frac{1}{\iota}b^TOb$, for any positive definite matrix $O\in \mathbb{R}^{n\times n} $, $a,b\in \mathbb{R}^{n}$ and $\iota>0 $, ones get
			\begin{eqnarray}
				\begin{split}
					\mathfrak{L}& V(\xi(t),t)\\
					\le &\eta_1^l \xi_l^T\!(t)[R(\tilde{C}_l(t)\tilde{\Sigma}_l(t))\otimes (A ^TP\!+\!PA \!+\!C^T\!PC )] \xi_l(t)\\
					&-\xi_l^T(t)\left[2\sqrt{\eta_3^l\check{\psi}_l}(\tilde{C}_l(t)\tilde{\Sigma}_l(t))\otimes \Gamma\right] \xi_l(t)\\
					&+\eta_1^f \xi_f^T\!(t)[S(\tilde{C}_f(t)\tilde{\Sigma}_f(t))\otimes (A ^TP\!+\!PA \!+\!C^T\!PC )] \xi_f(t)\\
					&-\xi_f^T(t)\left[2\sqrt{\eta_3^f\check{\psi}_f}(\tilde{C}_f(t)\tilde{\Sigma}_f(t))\otimes \Gamma\right] \xi_f(t).\notag
				\end{split}
			\end{eqnarray}
			Let $\check{\psi_i}\ge \frac{(\eta_1^l\max\{r_i\})^2}{4\eta_3^l}$ for $i=1,\dots, M$ and $\check{\psi_i}\ge \frac{(\eta_1^f\max\{s_{i-M}\})^2}{4\eta_3^f}$ for $i=M+1,\dots, N$. Then, we have
			\begin{eqnarray}\label{Sectionv2-6-linear}
				\begin{split}
					\mathfrak{L} V(\xi(t),t)\le& -\eta_1^l\xi_l^T(t)[R(\tilde{C}_l(t)\tilde{\Sigma}_l(t))\otimes I_n] \xi_l(t)\\ &-\eta_1^f\xi_f^T(t)[S(\tilde{C}_f(t)\tilde{\Sigma}_f(t))\otimes I_n] \xi_f(t)\\
					\le &-\eta_1^l\varepsilon \xi_l^T(t)[R(\tilde{C}_l(t)\tilde{\Sigma}_l(t))\otimes P] \xi_l(t)\\ &-\eta_1^f\varepsilon\xi_f^T(t)[S(\tilde{C}_f(t)\tilde{\Sigma}_f(t))\otimes P] \xi_f(t)\\
					= &-\eta_1^l\varepsilon \sum_{i=1}^{M}r_ic_i(t) \Sigma_i(t)\sigma_i(t)\\
					&-\eta_1^f\varepsilon \sum_{i=M+1}^{N}s_{i-M}c_i(t) \Sigma_i(t)\sigma_i(t)\\
					\le &-\eta_1^l\varepsilon \sum_{i=1}^{M}r_ic_i(t)\!\! \int_0^{\sigma_i(t)} \!\!\!k_1\Big(k_2+\frac{s}{c_i(t)}\Big)^{\mu}ds\\
					&-\eta_1^f\varepsilon \!\!\!\sum_{i=M+1}^{N}\!\!\!\!\!s_{i-M}c_i(t)\!\!\! \int_0^{\sigma_i(t)} \!\!\!\!\!\!k_1\Big(k_2\!+\!\frac{s}{c_i(t)}\Big)^{\mu}ds\\
					\le &-\varepsilon V(\xi(t),t)+\varphi(t),
				\end{split}
			\end{eqnarray}
			where $A ^TP+PA -PB  B  ^TP+C^TPC+I_n=0$, $\varepsilon=\frac{1}{\lambda_{\max}(P)}$, and $\varphi(t) =\eta_2^l \varepsilon \sum_{i=1}^{M}r_i(c_i(t)-\psi_i)^2+\eta_2^f \varepsilon \sum_{i=M+1}^{N}s_{i-M}(c_i(t)-\psi_i)^2$.
			Next, we will prove the existence and uniqueness, as well as m.s. and a.s. stability  of the solution to \eqref{Section4-3-linear} in the following two steps.
			
			Step 1:
			We need to prove the existence and uniqueness of the solution to the closed-loop stochastic system \eqref{Section4-3-linear}, as the system is inherently highly nonlinear and no longer satisfy the linear growth condition.
			Since $(I_N \otimes A  + (\tilde{C}(t)\tilde{\Sigma}(t)) \mathcal{L} \otimes B  \mathcal{K} )\vartheta(t) \le (I_N \otimes A  )\vartheta(t)+k_1\big[\big(\int_0^t \vartheta^T(s)(\mathcal{L}^T\mathcal{L} \otimes PB  B  ^TP)\vartheta(s)ds+C_1(0)\big)\big(k_2+\vartheta^T(t)(\mathcal{L}^T\mathcal{L}\otimes P) \vartheta(t) \big)^{\mu}\big]\mathcal{L} \otimes B  \mathcal{K} ]\vartheta(t)$ where $C_1(0)\ge 0$ holds, the drift and diffusion terms of closed-loop stochastic system \eqref{Section4-3-linear} satisfy the local Lipschitz condition.
			According to Lemma \ref{lemma-1}, the existence and uniqueness of the solution to \eqref{Section4-3-linear} can be proved.
			
			Step 2: By the functional $\mathrm{It\hat{o}} $ formula for $V(\xi(t),t)$ and the estimation \eqref{Sectionv2-6-linear} of $\mathfrak{L}V(\xi(t),t)\le -\eta_1^l\varepsilon \xi_l^T(t)[R(\tilde{C}_l(t)\tilde{\Sigma}_l(t))\otimes P] \xi_l(t)-\eta_1^f\varepsilon\xi_f^T(t)[S(\tilde{C}_f(t)\tilde{\Sigma}_f(t))\otimes P] \xi_f(t) \le-\eta_1\varepsilon\xi^T(t)[\tilde{S}(\tilde{C}(t)\tilde{\Sigma}(t))\otimes P] \xi(t)$ with $\eta_1=\eta_1^l \wedge \eta_1^f$ and $\tilde{S}=\mathrm{diag}\{R,S \}$, it can be deduced that
			\begin{eqnarray}\label{Sectionv2-S2-1}
				\begin{split}
					&V(\xi(t),t)-V(\xi(0),0)\\
					&=\int_0^t \mathfrak{L}V(\xi(s),s)ds+M_{v1}(t)\\
					&\le -\eta_1\varepsilon \!\int_0^t \!\xi^T(s)[\tilde{S}(\tilde{C}(s)\tilde{\Sigma}(s))\!\otimes\! P] \xi(s)ds+M_{v1}(t),
				\end{split}
			\end{eqnarray}
			where $M_{v1}(t)$ is a local martingale with $M_{v1}(0)=0$.
			Using the non-negative semi-martingale convergence theorem, ones have
			\begin{eqnarray}
				\begin{split}
					\lim_{t \to \infty}\sup V(\xi(t),t)<\infty \ a.s.\notag
				\end{split}
			\end{eqnarray}
			In addition, taking the expectations on both sides of \eqref{Sectionv2-S2-1}, we can obtain
			\begin{eqnarray}
				\begin{split}
					\ \int_0^t \!\mathrm{E}\eta_1 \varepsilon  \xi^T\!(s)[\tilde{S}(\tilde{C}(s)\tilde{\Sigma}(s))\!\otimes\! P] \xi(s)ds \le \! \mathrm{E}V(\xi(0),0) \!<\!\infty,\notag
				\end{split}
			\end{eqnarray}
			that is, letting $t \to \infty$ yields
			\begin{eqnarray}
				\begin{split}
					\int_0^\infty \mathrm{E}\eta_1 \varepsilon \xi^T(s)[\tilde{S}(\tilde{C}(s)\tilde{\Sigma}(s))\otimes P] \xi(s)ds <\infty.\notag
				\end{split}
			\end{eqnarray}
			Combining the above inequality with Barbalat's lemma \cite{krstic1995nonlinear}, we have $\lim_{t \to \infty} \mathrm{E} \xi^T(t)[\eta_1\varepsilon \tilde{S}(\tilde{C}(t)\tilde{\Sigma}(t))\otimes P] \xi(t)=0$. Note that $\tilde{C}(t)\tilde{\Sigma}(t) \ge \tilde{C}(0)>I_N$ and $\eta_1, \varepsilon>0$. Therefore, we have $\lim_{t \to \infty} \mathrm{E}\|\xi(t)\|^2=0$.
				Note that $\xi(t)=(\mathcal{L} \otimes I_n)\vartheta(t)$ and $\lim_{|\xi| \to \infty}\inf_{t\ge 0} V(\xi,t)=\infty$. Using the stochastic LaSalle lemma \cite{wang2020advances}, we can get $\lim_{t \to \infty} \xi(t)=0$ a.s. According to the definition of $\xi(t)$, the m.s. and a.s. consensus can be solved.
		\end{proof}
		
		\begin{remark}\label{remark-11}
			In fact, the function $\Sigma_i$ admits other design forms. Especially, we can choose $\Sigma_i(t)= k_1(k_2+\sigma_i(t))^{\mu}$ with $\mu>1, k_1, k_2\ge 1$, $c_i(t)\ge 1$, $\Sigma_i(t)\ge 1$. By selecting the Lyapunov function $\tilde{V}(\xi(t),t)$ $=V(\xi(t),t)$, it can be deduced
			that $\mathfrak{L} \tilde{V}(\xi(t),t)= \sum_{i=1}^{M}r_i[\eta_1^l\Sigma_i(t)\sigma_i(t)+2\eta_2^lc_i(t)-2\eta_2^l\psi_i]\dot{c}_i(t)+\eta_1^l\sum_{i=1}^{M}r_i c_i(t)\Sigma_i(t)\mathfrak{L}\sigma_i(t)+\sum_{i=M+1}^{N}s_{i-M}[\eta_1^f\Sigma_i(t)\sigma_i(t)+2\eta_2^fc_i(t)-2\eta_2^f\psi_i]\dot{c}_i(t)+\eta_1^f\sum_{i=M+1}^{N}s_{i-M} c_i(t)\Sigma_i(t)\mathfrak{L}\sigma_i(t)$.
			By the similar steps of the above proof of Theorem \ref{theorem-1}, the m.s. and a.s. consensus of \eqref{2-1} can be solved by the protocol \eqref{Section4-1} with $\Sigma_i(t)= k_1(k_2+\sigma_i(t))^{\mu}$, for $\mu>1, k_1, k_2\ge 1$ and $c_i(0)\ge 1$, $\Sigma_i(0)\ge 1$.
			In particular, the feedback control structure proposed in \cite{li2015designing} is recovered as a special case of this protocol by setting $\mu = 3$ and $k_1 = k_2 = 1$.
			
		\end{remark}
		
		Furthermore, for the case of $\mu=1$ and $\gamma=0$, the control protocol \eqref{Section4-1} can be written as
		\begin{eqnarray}\label{4-1}
			\begin{split}
				&u_i(t)=k_1(k_2c_i(t)+\sigma_i(t))\mathcal{K} \xi_i(t),\\
				&\dot{c}_i(t)=\xi_i^T(t)\Gamma \xi_i(t),
			\end{split}
		\end{eqnarray}
		where $k_1, k_2>0$, $\xi_i(t)=\sum_{j\in N_i} a_{ij}(x_i(t)-x _j(t))$ and $\sigma_i(t)=\xi_i^T(t)P\xi_i^T(t)$. Then we have the following theorem.
		
		\begin{theorem}\label{the4-1-linear}
			Suppose the directed graph $\tilde{\mathcal{G}}$ contains a spanning tree. The m.s. and a.s. consensus of \eqref{2-1} can be solved by the protocol \eqref{4-1} with $\mathcal{\mathcal{K} } =-B  ^TP$, and $\Gamma=PB  B  ^TP$, where $P$ is the solution to the SARE \eqref{2-2}.
		\end{theorem}
		
		\begin{proof}
			Substituting the control protocol \eqref{4-1} into \eqref{2-1} yields the following closed-loop system like \eqref{Section4-3-linear}
			\begin{eqnarray}
				\begin{split}\label{4-3-linear}
					d\vartheta(t)=&(I_N \otimes A  +  k_1(k_2\tilde{C}(t)+\tilde{\sigma}(t)) \mathcal{L} \otimes B  \mathcal{K} )\vartheta(t)dt\\&+(I_N \otimes C)\vartheta(t)dw(t).
				\end{split}
			\end{eqnarray}
			Then, according to the graph decomposition in Lemma \ref{lem-decomposition}, we have
			\begin{eqnarray}\label{xi-l-1}
				\begin{split}
					d\xi_l(t)\!=&(I_M \!\otimes\! A  \!+\!  (k_1(k_2\tilde{C}_l(t)\!+\!\tilde{\sigma}_l(t)) \mathcal{L}_{11}) \!\otimes\! B  \mathcal{K} )\xi_l(t)dt\!\!\\&\!+\!(I_M \!\otimes\! C)\xi_l(t)dw(t),
				\end{split}
			\end{eqnarray}
			and
			\begin{eqnarray}\label{xi-f-1}
				\begin{split}
					d\xi_{\!f\!}(t)\!=&(I_{N\!-\!M} \!\otimes\! A  \!+\!  (k_1(k_2\tilde{C}_{\!f\!}(t)\!+\!\tilde{\sigma}_{\!f\!}(t))\mathcal{L}_{22}) \!\otimes\! B  \mathcal{K} )\xi_{\!f\!}(t)dt\!\!\!\!\!\!\!\\&+  (k_1(k_2\tilde{C}_l(t)\!+\!\tilde{\sigma}_l(t))\mathcal{L}_{21} \otimes B  \mathcal{K} )\xi_l(t)dt\\
					&+(I_{N\!-\!M} \!\otimes \!C)\xi_f(t)dw(t).
				\end{split}
			\end{eqnarray}
			Note that the Lyapunov functional $V(\xi(t), t)$ in Theorem \ref{theorem-1} satisfies $V(\xi(t), t)\le \sum_{i=1}^{M}r_i[\eta_1^lk_1(2k_2c_i(t)+\sigma_i(t))\sigma_i(t)+\eta_2^l(c_i(t)-\psi_i)^2]+\sum_{i=M+1}^{N}s_{i-M} [\eta_1^fk_1(2k_2c_i(t)+\sigma_i(t))\sigma_i(t)+\eta_2^f(c_i(t)-\psi_i)^2]$ with $\mu=1$. Then, we choose Lyapunov functional
			\begin{eqnarray}
				\begin{split}
					V_1&(\xi(t),t)\\
					=&\sum_{i=1}^{M}\!r_i\!\Big[\eta_1^lk_1(2k_2c_i(t)\!+\!\sigma_i(t))\sigma_i(t)\!+\!\eta_2^l(c_i(t)\!-\!\psi_i)^2\Big]\\
					&\!\!+\!\!\!\!\sum_{i=M+1}^{N}\!\!\!\!\!s_{i-M}\!\Big[\eta_1^fk_1(2k_2c_i(t)\!+\!\sigma_i(t))\sigma_i(t)\!+\!\eta_2^f(c_i(t)\!-\!\psi_i)^2\Big],\notag
				\end{split}
			\end{eqnarray}
			where $\eta_1^l, \eta_2^l, \eta_1^f, \eta_2^f>0$, $\psi_i$ are positive constants to be determined. According to the functional $\mathrm{It\hat{o} } $ formula, it can be deduced that $d V_1(\xi(t),t)$ has the following form
			\begin{eqnarray}
				\begin{split}
					dV_{\!1}&(\xi(t),t)\!=\!\mathfrak{L}V_{\!1}(\xi(t),t) dt\!+\! \frac{\partial V_1(\xi(t),t)}{\partial \xi}\!(I_N\!\otimes\! C   )\xi(t)dw(t),\notag
				\end{split}
			\end{eqnarray}
			with 
			\begin{eqnarray}
					\begin{split}
							\mathfrak{L}& V_1(\xi(t),t)\\
				= & \sum_{i=1}^{M}\!r_i\big[(\eta_{k_2}^l\sigma_i(t)\!+\!2\eta_2^lc_i(t)\!-\!2\eta_2^l\psi_i)\dot{c}_i(t)\big]\\
				&+\sum_{i=1}^{M}\!r_i\big[\big(\eta_{k_2}^lc_i(t)\!+\!\eta_{k_1}^l\sigma_i(t)\big)\mathfrak{L}\sigma_i(t)\big]\\
				&+\sum_{i=M+1}^{N}\!s_{i-M}\big[(\eta_{k_2}^f\sigma_i(t)\!+\!2\eta_2^fc_i(t)\!-\!2\eta_2^f\psi_i)\dot{c}_i(t)\big]\\
				&+\sum_{i=M+1}^{N}\!s_{i-M}\big[\big(\eta_{k_2}^fc_i(t)\!+\!\eta_{k_1}^f\sigma_i(t)\big)\mathfrak{L}\sigma_i(t)\big],\notag
						\end{split}
				\end{eqnarray}
		with $\eta_{k_1}^l=2\eta_1^lk_1$, $\eta_{k_2}^l=2\eta_1^lk_1k_2$, $\eta_{k_1}^f=2\eta_1^fk_1$, and $\eta_{k_2}^f=2\eta_1^fk_1k_2$. Denote $\tilde{\psi}=\mathrm{diag}\{\psi_1,...,\psi_N\}$. According to Lemma \ref{lem-decomposition}, $\tilde{\psi}$ can be partitioned into $\tilde{\psi}_l=\mathrm{diag}\{\psi_1,...,\psi_M\}$ and $\tilde{\psi}_f=\mathrm{diag}\{\psi_{M+1},...,\psi_N\}$. Let $\mathcal{K} =-B  ^TP$ and $\Gamma=PB  B  ^TP$.
		Note that $\sum_{i=1}^{M}\sigma_i(t)=\xi_l^T(t)(I_{M}\otimes P)\xi_l(t)$ and $\xi_l(t)=(\mathcal{L}_{11} \otimes I_n)\vartheta_l(t)$. By the definition of $\sigma_i(t)$, \eqref{xi-l-1}, \eqref{xi-f-1} and using Lemma \ref{lem-decomposition} and \ref{lem-strongly-connected}, we have
		\begin{eqnarray}
			\begin{split}
				&\sum_{i=1}^{M}\!r_i\big[\big(\eta_{k_2}^lc_i(t)\!+\!\eta_{k_1}^l\sigma_i(t)\big)\mathfrak{L}\sigma_i(t)\big]\\
				&
				= \xi_l^T\!(t)[R(\eta_{k_2}^l\tilde{C}_l(t)\!+\!\eta_{k_1}^l\tilde{\sigma}_l(t))\!\otimes\! (A ^T\!P\!+\!PA \!+\!C^T\!PC   )] \xi_l(t)\!\!\!\!\\ &\ \ \ -\!2\eta_1^l\xi_l^T(t)\big[(k_1k_2\tilde{C}_l(t)+k_1\tilde{\sigma}_l(t))^2\tilde{\mathcal{L}}_{11}\otimes \Gamma\big] \xi_l(t),\notag
			\end{split}
		\end{eqnarray}
		and
		\begin{eqnarray}
			\begin{split}
				&\sum_{i=1}^{M}r_i(\eta_{k_2}^l\sigma_i(t)+2\eta_2^lc_i(t)-2\eta_2^l\psi_i)\dot{c}_i(t)\\
				&=\xi_l^T(t)[R(\eta_{k_2}^l\tilde{\sigma}_l(t)\!+\!2\eta_2^l\tilde{C}_l(t)\!-\!2\eta_2^l\tilde{\psi}_l)\otimes \Gamma] \xi_l(t).\notag
			\end{split}
		\end{eqnarray}
		Note that $\sum_{i=M+1}^{N}\sigma_i(t)=\xi_f^T(t)(I_{N-M}\otimes P)\xi_f(t)$. Then, using Lemma \ref{lem-decomposition}, \ref{lem-nonsingular-M}, and the elementary inequality: $2a^TOb\le \iota a^TOa+\frac{1}{\iota}b^TOb$, for any positive definite matrix $O\in \mathbb{R}^{n\times n} $, $a,b\in \mathbb{R}^{n}$ and $\iota>0 $, we have
		\begin{eqnarray}\label{V1-f-1}
			\begin{split}
				&\sum_{i=M+1}^{N}\!s_{i-M}\big[\big(\eta_{k_2}^fc_i(t)\!+\!\eta_{k_1}^f\sigma_i(t)\big)\mathfrak{L}\sigma_i(t)\big]\\
				&
				=\! \xi_f^T\!(t)[S(\eta_{k_2}^f\!\tilde{C}_f(t)\!+\!\eta_{k_1}^f\!\tilde{\sigma}_f(t))\!\otimes\! (A ^T\!P\!+\!PA \!+\!C^T\!PC   )] \xi_f(t)\!\!\!\!\\ &\ \ \ -2\eta_1^f\xi_f^T(t)\big[(k_1k_2\tilde{C}_f(t)+k_1\tilde{\sigma}_f(t))^2\tilde{\mathcal{L}}_{22}\otimes \Gamma\big] \xi_f(t)\\ &\ \ \ -\!2\eta_1^f\xi_f^T\!(t)\big[(k_1k_2\tilde{C}_l(t)\!+\!k_1\tilde{\sigma}_l(t))\\
				&\quad \quad \quad \quad \quad \quad \cdot (k_1k_2\tilde{C}_f(t)\!+\!k_1\tilde{\sigma}_f(t))S\mathcal{L}_{21}\!\otimes \! \Gamma\big] \xi_l(t)\\
				&
				\le \! \xi_f^T\!(t)[S(\eta_{k_2}^f\!\tilde{C}_f(t)\!+\!\eta_{k_1}^f\!\tilde{\sigma}_f(t))\!\otimes\! (A ^T\!P\!+\!PA \!+\!C^T\!PC   )] \xi_f(t)\!\!\!\!\\ &\ \ \ -2\eta_1^f\xi_f^T(t)\big[(k_1k_2\tilde{C}_f(t)+k_1\tilde{\sigma}_f(t))^2\tilde{\mathcal{L}}_{22}\otimes \Gamma\big] \xi_f(t)\\ &\ \ \ +2\eta_1^f\iota\xi_f^T\!(t)[(k_1k_2\tilde{C}_f(t)\!+\!k_1\tilde{\sigma}_f(t))^2\!\otimes \! \Gamma] \xi_f(t)\\ &\ \ \ +2\eta_1^f\frac{\sigma_{\max}^2(S\mathcal{L}_{21})}{\iota}\xi_l^T\!(t)[(k_1k_2\tilde{C}_l(t)\!+\!k_1\tilde{\sigma}_l(t))^2\!\otimes \! \Gamma] \xi_l(t),\notag
			\end{split}
		\end{eqnarray}
		and
		\begin{eqnarray}\label{V1-f-2}
			\begin{split}
				&\sum_{i=M+1}^{N}\!s_{i-M}\big[(\eta_{k_2}^f\sigma_i(t)\!+\!2\eta_2^fc_i(t)\!-\!2\eta_2^f\psi_i)\dot{c}_i(t)\big]\\
				&=\xi_f^T(t)[S(\eta_{k_2}^f\tilde{\sigma}_f(t)\!+\!2\eta_2^f\tilde{C}_f(t)\!-\!2\eta_2^f\tilde{\psi}_f)\otimes \Gamma] \xi_l(t).\notag
			\end{split}
		\end{eqnarray}
		Denote $\Theta^l_t=\eta_{k_2}^l\tilde{C}_l(t)+\eta_{k_1}^l\tilde{\sigma}_l(t)$ and $\Theta^f_t=\eta_{k_2}^f\tilde{C}_f(t)+\eta_{k_1}^f\!\tilde{\sigma}_f(t)$. Substituting the above inequalities into $\mathfrak{L}V_1(\xi(t),t)$ yields
		\begin{eqnarray}
			\begin{split}
				\mathfrak{L}& V_1(\xi(t),t)\\
				\le&\xi_l^T\!(t)[R\Theta^l_t\!\otimes\! (A ^T\!P\!+\!PA \!+\!C^T\!PC   )] \xi_l(t)\\ &+\xi_f^T\!(t)[S\Theta^f_t\!\otimes\! (A ^T\!P\!+\!PA \!+\!C^T\!PC   )] \xi_f(t)\\&-2\tilde{\eta}_l\xi_l^T\!(t)\big[(k_1k_2\tilde{C}_l(t)\!+\!k_1\tilde{\sigma}_l(t))^2\!\otimes \! \Gamma\big] \xi_l(t)\\ &-2\tilde{\eta}_f\xi_f^T(t)\big[(k_1k_2\tilde{C}_f(t)+k_1\tilde{\sigma}_f(t))^2\otimes \Gamma\big] \xi_f(t)\\ & +\xi_l^T(t)[R(\eta_{k_2}^l\tilde{\sigma}_l(t)\!+\!2\eta_2^l\tilde{C}_l(t)\!-\!2\eta_2^l\tilde{\psi}_l)\!\otimes\! \Gamma] \xi_l(t)\\
				&+\xi_f^T(t)[S(\eta_{k_2}^f\tilde{\sigma}_f(t)\!+\!2\eta_2^f\tilde{C}_f(t)\!-\!2\eta_2^f\tilde{\psi}_f)\otimes \Gamma] \xi_l(t),\notag
			\end{split}
		\end{eqnarray}
		where $\tilde{\eta}_l=\eta_1^l \lambda_2(\tilde{\mathcal{L}}_{11})-\eta_1^f\frac{\sigma_{\max}^2(S\mathcal{L}_{21})}{\iota}$ and $\tilde{\eta}_f=\eta_1^f( \lambda_2(\tilde{\mathcal{L}}_{22})-\iota)$. Choose parameters $\eta_1^l,\eta_1^f,\iota $ to ensure $\tilde{\eta}_l>0$ and $\tilde{\eta}_f>0$. Let $\eta_2^l=\eta_1^l k_1k_2^2$ and $\eta_2^f=\eta_1^f k_1k_2^2$. Using the elementary inequality: $2a^TOb\le \iota a^TOa+\frac{1}{\iota}b^TOb$, for any positive definite matrix $O\in \mathbb{R}^{n\times n} $, $a,b\in \mathbb{R}^{n}$ and $\iota>0 $, ones get
		\begin{eqnarray}
			\begin{split}
				\mathfrak{L}& V_1(\xi(t),t)\\
				\le&\xi_l^T\!(t)[R\Theta^l_t\!\otimes\! (A ^T\!P\!+\!PA \!+\!C^T\!PC \!+\!k_2\Gamma  )] \xi_l(t)\\ &\!\!\!\!+\!\xi_f^T\!(t)[S\Theta^f_t\!\otimes\! (A ^T\!P\!+\!PA \!+\!C^T\!PC\!+\!k_2\Gamma)] \xi_f(t)\\&\!\!\!\!-\!\xi_l^T\!(t)\big[(2\tilde{\eta}_l(k_1k_2\tilde{C}_l(t)\!+\!k_1\tilde{\sigma}_l(t))^2\!+\!2\eta_2^l\tilde{\psi}_lR)\!\otimes \! \Gamma\big] \xi_l(t)\\ &\!\!\!\!-\!\xi_f^T(t)\big[(2\tilde{\eta}_f(k_1k_2\tilde{C}_f(t)\!+\!k_1\tilde{\sigma}_f(t))^2\!+\!2\eta_2^f\tilde{\psi}_fS)\otimes \Gamma\big] \xi_f(t)\\
				\le&\xi_l^T\!(t)[R\Theta^l_t\!\otimes\! (A ^T\!P\!+\!PA \!+\!C^T\!PC \!+\!k_2\Gamma  )] \xi_l(t)\\ &\!\!\!\!+\!\xi_f^T\!(t)[S\Theta^f_t\!\otimes\! (A ^T\!P\!+\!PA \!+\!C^T\!PC\!+\!k_2\Gamma)] \xi_f(t)\\&\!\!\!\!-\!2\xi_l^T\!(t)\Big[\sqrt{\tilde{\eta}_l \eta_2^l\tilde{\psi}_lR}\Big(k_1k_2\tilde{C}_l(t)\!+\!k_1\tilde{\sigma}_l(t)\Big)\!\otimes \! \Gamma\big] \xi_l(t)\\ &\!\!\!\!-\!2\xi_f^T(t)\big[\sqrt{\tilde{\eta}_f\eta_2^f\tilde{\psi}_fS}\Big(k_1k_2\tilde{C}_f(t)+k_1\tilde{\sigma}_f(t)\Big)\otimes \Gamma\Big] \xi_f(t).\notag
			\end{split}
		\end{eqnarray}
		Let $\psi_i\ge \frac{(\eta_1^l)^2\max\{r_i\}(1+k_2)^2}{\tilde{\eta}_l \eta_2^l}$ for $i=1,\dots, M$ and $\psi_i\ge \frac{(\eta_1^f)^2\max\{s_{i-M}\}(1+k_2)^2}{\tilde{\eta}_f \eta_2^f}$ for $i=M+1,\dots, N$. Then, we have
		\begin{eqnarray}\label{v2-6-linear}
			\begin{split}
				\mathfrak{L} V_1(\xi(t),t)\le& -\!\xi_l^T\!(t)[R\Theta^l_t\otimes I_n] \xi_l(t)\! -\!\xi_f^T\!(t)[S\Theta^f_t\!\otimes\! I_n] \xi_f(t)\\
				\le &-\!\varepsilon \xi_l^T\!(t)\Big[R\Big(\eta_{k_2}^l\tilde{C}_l(t)\!+\!\frac{\eta_{k_1}^l}{2}\tilde{\sigma}_l(t)\Big)\otimes P\Big] \xi_l(t)\\& -\!\varepsilon\xi_f^T\!(t)\Big[S\Big(\eta_{k_2}^f\tilde{C}_f(t)\!+\!\frac{\eta_{k_1}^f}{2}\tilde{\sigma}_f(t)\Big)\otimes P\Big] \xi_f(t)\\
				\le &-\varepsilon V_1(\xi(t),t)+\varphi(t),\notag
			\end{split}
		\end{eqnarray}
		where $A ^TP+PA -PB  B  ^TP+C^TPC +I_n=0$, $\varepsilon=\frac{1}{\lambda_{\max}(P)}$, and $\varphi(t) =\eta_2^l \varepsilon \sum_{i=1}^{M}r_i(c_i(t)-\psi_i)^2+\eta_2^f \varepsilon \sum_{i=M+1}^{N}s_{i-M}(c_i(t)-\psi_i)^2$.
		Then, according to Lemma \ref{lemma-1} and the proof of Theorem \ref{theorem-1}, we can prove the existence and uniqueness, as well as m.s. and a.s. stability  of the solution to \eqref{4-3-linear}. Then, by the definition of $\xi(t)$, the m.s. and a.s. consensus can be solved.
	\end{proof}
	
	\begin{remark}
		Theorems \ref{theorem-1} and \ref{the4-1-linear} established a unified and flexible framework for the design of fully distributed protocol in stochastic environments under directed topology. 
		In contrast to the existing adaptive protocols \cite{gu2020adaptive,wen2024fully}, which are designed for stochastic MASs under undirected graphs, the proposed protocols address the case of directed graphs and provide a fully distributed control framework that can be applied to solve both the m.s. and a.s. consensus under directed graphs.
		To this end, the adaptive strategy \eqref{Section4-1} incorporates two time-varying gains, $c_i(t)$ and $\Sigma_i(t)$, each serving a distinct role in facilitating stochastic consensus: $c_i(t)$ is employed to estimate global information, while $\Sigma_i(t)$ accounts for the asymmetrical interactions inherent in directed graphs. 
		
%
	\end{remark}


	\begin{remark}
		In fact, the structure of the auxiliary time-varying gain function $\Sigma_i(t)$, parameterized by $k_1$, $k_2$, $\mu$, and the time-varying terms $c_i(t)$, $\sigma_i(t)$, introduces generality into the protocol design. This generality allows our adaptive protocol to encompass several representative consensus designs for deterministic multi-agent systems under directed topologies as special cases, thereby demonstrating broader applicability. For instance, the adaptive gain design in \cite{zhou2024fully} can be obtained from Theorem \ref{the4-1-linear} by setting $k_1 = k_2 = 1$, which can be regarded as specific parameterizations within our proposed framework \eqref{Section4-1}.
%
	\end{remark}
	
	\subsection{Fully distributed consensus control for stochastic MASs with the undirected topology}\label{sec3}
	The parameter $\Sigma_i(t)$ in the protocol \eqref{Section4-1} is introduced to account for asymmetric interactions in the network;  however, since undirected graphs exhibit symmetric adjacency relations by definition, $\Sigma_i(t)$ becomes superfluous and can be safely omitted in such case. Assume that the undirected graph $\bar{\mathcal{G}}$ is connected.
	Therefore, the eigenvalues of $\mathcal{L}$ for the undirected graph $\bar{\mathcal{G}}$ are real and nonnegative, denoted by  $\lambda_1(\mathcal{L})=0$ and $0<\lambda_2(\mathcal{L})\leq\ldots\leq\lambda_N(\mathcal{L})$.
	
	Then, we have the following theorem for the case of undirected topology.
	\begin{theorem}\label{the-3.1}
		Suppose the undirected graph $\bar{\mathcal{G}}$ is connected. The m.s. and a.s. consensus of \eqref{2-1} can be solved by the protocol \eqref{Section4-1} with $\Sigma_i(t)=1$, $\gamma=0$, $\mathcal{K} =-B  ^TP$, and $\Gamma=PB  B  ^TP$, where $P$ is the solution to the SARE \eqref{2-2}.
	\end{theorem}
	\begin{proof}
	Let $\Sigma_i(t)=1$. As mentioned in Theorem \ref{theorem-1}, the closed-loop MASs can be written as follows
	\begin{eqnarray}
		\begin{split}
			d\vartheta(t)=&(I_N \otimes A  +  \tilde{C}(t) \mathcal{L} \otimes B  \mathcal{K} )\vartheta(t)dt\\&+(I_N \otimes C)\vartheta(t)dw(t).
		\end{split}
	\end{eqnarray}
	Then, we choose Lyapunov function
	\begin{eqnarray}\label{V2}
		\begin{split}
			V_2(\vartheta(t),t)=\vartheta^T(\mathcal{L}\otimes P)\vartheta+\frac{1}{2}\sum_{i=1}^{N}(c_i(t)-\psi_i)^2,
		\end{split}
	\end{eqnarray}
	where $P>0$ and $\psi_i>0$. Applying the functional $\mathrm{It\hat{o}} $  formula to \eqref{V2}, we have
	\begin{eqnarray}
		\begin{split}
			dV_2&(\vartheta(t),t)\!=\!\mathfrak{L}V_2(\vartheta(t),t) dt+ 2\vartheta^T\!(t)(\mathcal{L}\otimes P C   )\vartheta(t)dw(t),\notag
		\end{split}
	\end{eqnarray}
	where $\mathfrak{L}V_2(\vartheta(t),t)$ is defined as
	\begin{eqnarray}
		\begin{split}
			\mathfrak{L}V_2&(\vartheta(t),t)=\vartheta^T(t)[\mathcal{L}\otimes (A ^TP+PA+C^TPC )] \vartheta(t)\\&+2\vartheta^T(t)[\mathcal{L}\tilde{C}(t)\mathcal{L}\otimes PB  \mathcal{K} ] \vartheta(t)+(c_i(t)-\psi_i)\dot{c}_i(t).\notag
		\end{split}
	\end{eqnarray}
	Let $\mathcal{K} =-B  ^TP$ and $\Gamma=PB  B  ^TP$. Then, we can obtain
	\begin{eqnarray}
		\begin{split}
			\mathfrak{L}V_2&(\vartheta(t),t)= \vartheta^T(t)[\mathcal{L}\otimes (A ^T\!P\!+\!PA \!+\!C^T\!PC )] \vartheta(t)\\&-2\vartheta^T(t)[\mathcal{L}\tilde{C}(t)\mathcal{L}\otimes PB  B  ^TP] \vartheta(t)\\&+\sum_{i=1}^{N}(c_i(t)-\psi_i)\vartheta_i^T(t)PB  B  ^TP \vartheta_i(t).\notag
		\end{split}
	\end{eqnarray}
	Note that $\vartheta_i(t)=\sum_{j\in N_i} (x_i(t)-x _j(t))=\sum_{j\in N_i} (\vartheta_i(t)-\vartheta_j(t))$. Then, we have
	\begin{eqnarray}\label{v1-2-linear}
		\begin{split}
			&\sum_{i=1}^{N}(c_i(t)-\psi_i)\vartheta_i^T(t)PB  B  ^TP \vartheta_i(t)\\&=\vartheta^T(t)[\mathcal{L}(\tilde{C}(t)-\tilde{\psi})\mathcal{L}\otimes PB  B  ^TP] \vartheta(t),
		\end{split}
	\end{eqnarray}
	where $\tilde{\psi}=\mathrm{diag}\{\psi_1,...,\psi_N\}$.
	Denote $\varepsilon=\frac{1}{\lambda_{\max}(P)}$. Let $\psi_i>\frac{1}{\lambda_2(\mathcal{L})}$. Note that $A ^TP+PA -PB  B  ^TP+C^TPC+I_n=0$. Thus, Substituting \eqref{v1-2-linear} into $\mathfrak{L}V_2(\vartheta(t),t)$ yields
	\begin{eqnarray}
		\begin{split}
			\mathfrak{L}&V_2(\vartheta(t),t)\\
			\le &\vartheta^T(t)[\mathcal{L}\otimes (A ^TP+PA -\tilde{\psi}\lambda_2(\mathcal{L}) PB  B  ^TP)] \vartheta(t)\\&+\vartheta^T(t)[\mathcal{L}\otimes (C^TPC  )] \vartheta(t)\\
			\le &-\vartheta^T(t)[\mathcal{L}\otimes I_n] \vartheta(t)\\
			\le &-\varepsilon V_2(\vartheta(t),t)+\bar{\varphi}(t),
		\end{split}
	\end{eqnarray}
	where $\bar{\varphi}(t)=\frac{\varepsilon}{2}\sum_{i=1}^{N}(c_i(t)-\psi_i)^2$.
	By Lemma \ref{lemma-1} and the similar steps of Theorem \ref{theorem-1}, the m.s. and a.s. consensus of \eqref{2-1} can be solved by the protocol \eqref{Section4-1} with $\Sigma_i(t)=1$ for the case of undirected topology.
\end{proof}
	Although the directed graph protocol discussed with $\gamma=0$ in the above theorem is theoretically applicable to undirected graphs, we exploit the advantageous structural properties of undirected graphs to develop an enhanced distributed control protocol.  In the following, we consider a fully distributed control protocol with tunable parameters $\gamma$ specifically tailored for undirected graphs, which demonstrates superior convergence performance compared to the previously established protocol for the case of directed graphs.
	
	The control protocol \eqref{Section4-1} can be written as
	\begin{eqnarray}\label{3-1}
		\begin{split}
			&u_i(t)=c_i(t)\mathcal{K} \xi_i(t),\\
			&\dot{c}_i(t)=e^{\gamma t}\xi_i^T(t)\Gamma \xi_i(t),
		\end{split}
	\end{eqnarray}
	where $\xi_i(t)=\sum_{j\in N_i} a_{ij}(x_i(t)-x _j(t))$, $c_i(t)$ is an adaptive gain functinon, $\mathcal{K} \in \mathbb{R}^{m\times n}$ and $\Gamma \in \mathbb{R}^{n\times n}$ are the feedback gain matrices, and  $\gamma$ is a non-negative constant to be determined. Under the assumption that the undirected graph $\bar{\mathcal{G}}$ is connected, we can show that  the m.s. and a.s. consensus can be solved by \eqref{3-1}.
	
	\begin{theorem}\label{the3-1-linear}
		Suppose the undirected graph $\bar{\mathcal{G}}$ is connected. The m.s. and a.s. consensus of \eqref{2-1} can be solved by the protocol \eqref{3-1} with $\mathcal{K} =-B  ^TP$ and $\Gamma=PB  B  ^TP$, where $P$ is the solution to the SARE \eqref{2-2} and $\frac{1}{\lambda_{\max}(P)}\le \gamma< \frac{3}{2\lambda_{\max}(P)}$.
	\end{theorem}
	
	\begin{proof}
		Substituting the control protocol \eqref{3-1} into \eqref{2-1} yields the following closed-loop system
		\begin{eqnarray}\label{3-2-linear}
			\begin{split}
				dx  (t)=&(I_N \otimes A  + \tilde{C}(t) \mathcal{L} \otimes B  \mathcal{K} )x  (t)dt\\&+(I_N \otimes C)x(t)dw(t),
			\end{split}
		\end{eqnarray}
		where $\tilde{C}(t)=\mathrm{diag}\{c_1(t),...,c_N(t)\}$.
		Let $\vartheta_i(t)=x_i(t)-\frac{1}{N}\sum_{j=1}^{N} x _j(t)$.
		Denote $\vartheta(t)=[(I_N-\frac{1}{N}\mathbf{1}_N\mathbf{1}_N^T)\otimes I_n]x  (t)=[\vartheta_1^T(t),...,\vartheta_N^T(t)]^T$. Then, the closed-loop system \eqref{3-2-linear} can be written as follows
		\begin{eqnarray}
			\begin{split}\label{3-3-linear}
				d\vartheta(t)=&(I_N \otimes A  + \tilde{C}(t) \mathcal{L} \otimes B  \mathcal{K} )\vartheta(t)dt\\&+(I_N \otimes C)\vartheta(t)dw(t).
			\end{split}
		\end{eqnarray}
		Firstly, we choose the Lyapunov function
		\begin{eqnarray}\label{v1-1-linear}
			\begin{split}
				V_3(\vartheta(t),t)=\vartheta^T(\mathcal{L}\otimes P)\vartheta+\sum_{i=1}^{N}e^{-\gamma t}(c_i(t)-\psi_i)^2,
			\end{split}
		\end{eqnarray}
		where $P>0$ and $\psi_i>0$. Applying the functional $\mathrm{It\hat{o}} $  formula to \eqref{v1-1-linear}, we have
		\begin{eqnarray}
			\begin{split}
				dV_3&(\vartheta(t),t)=\mathfrak{L}V_3(\vartheta(t),t) dt+ 2\vartheta^T(\mathcal{L}\otimes P C   )\vartheta(t)dw(t),\notag
			\end{split}
		\end{eqnarray}
		where $\mathfrak{L}V_3(\vartheta(t),t)$ is defined as
		\begin{eqnarray}
			\begin{split}
				\mathfrak{L}V_3&(\vartheta(t),t)=\vartheta^T(t)[\mathcal{L}\otimes (A ^TP+PA+C^TPC)] \vartheta(t)\\&+2\vartheta^T(t)[\mathcal{L}\tilde{C}(t)\mathcal{L}\otimes PB  \mathcal{K} ] \vartheta(t)\\&+2\sum_{i=1}^{N}e^{-\gamma t}(c_i(t)-\psi_i)\dot{c}_i(t)-\gamma\sum_{i=1}^{N}e^{-\gamma t}(c_i(t)-\psi_i)^2.\notag
			\end{split}
		\end{eqnarray}
		Let $\mathcal{K} =-B  ^TP$ and $\Gamma=PB  B  ^TP$. Then, we can obtain
		\begin{eqnarray}
			\begin{split}
				\mathfrak{L}V_3&(\vartheta(t),t)=\vartheta^T(t)[\mathcal{L}\otimes (A ^T\!P\!+\!PA \!+\!C^T\!PC   )] \vartheta(t)\\&-2\vartheta^T(t)[\mathcal{L}\tilde{C}(t)\mathcal{L}\otimes PB  B  ^TP] \vartheta(t)\\&+2\sum_{i=1}^{N}(c_i(t)-\psi_i)\vartheta_i^T(t)PB  B  ^TP \vartheta_i(t)\\&-\gamma\sum_{i=1}^{N}e^{-\gamma t}(c_i(t)-\psi_i)^2.\notag
			\end{split}
		\end{eqnarray}
		Denote $\varepsilon=\frac{1}{\lambda_{\max}(P)}$. Let $\psi_i>\frac{1}{2\lambda_2(\mathcal{L})}$ and $\gamma\ge \varepsilon$. Note that $A ^TP+PA -PB  B  ^TP+C^TPC   +I_n=0$. Thus, Substituting \eqref{v1-2-linear} into $\mathfrak{L}V_3(\vartheta(t),t)$ yields
		\begin{eqnarray}\label{v1-5-linear}
			\begin{split}
				\mathfrak{L}&V_3(\vartheta(t),t)\\
				\le &\vartheta^T\!(t)[\mathcal{L}\otimes (A ^TP+PA -2\tilde{\psi}\lambda_2(\mathcal{L}) PB  B  ^TP)] \vartheta(t)\\&+\!\vartheta^T\!(t)[\mathcal{L}\!\otimes\! (C^T\!PC     )] \vartheta(t)\!-\!\gamma\sum_{i=1}^{N}\!e^{-\gamma t}(c_i(t)\!-\!\psi_i)^2\\
				\le &-\vartheta^T(t)[\mathcal{L}\otimes I_n] \vartheta(t)-\gamma\sum_{i=1}^{N}e^{-\gamma t}(c_i(t)-\psi_i)^2\\
				\le &-\varepsilon V_3(\vartheta(t),t).
			\end{split}
		\end{eqnarray}
		Next, we will prove the existence and uniqueness, m.s. stability, and a.s. stability  of the solution to \eqref{3-3-linear} in the following three steps.
		
		Step 1:
		Similar to Step 1 in the proof of Theorem \ref{theorem-1}, it can be seen that the drift and diffusion terms of closed-loop stochastic system \eqref{3-3-linear} satisfy the local Lipschitz condition.
		According to Lemma \ref{lemma-1}, the existence and uniqueness of the solution to \eqref{3-3-linear} can be proved.
		
		Step 2: By the functional $\mathrm{It\hat{o}} $ formula for $V_3(\vartheta(t),t)$, we have
		\begin{eqnarray}
			\begin{split}
				&\mathrm{E} (e^{\delta(\tau_h \wedge t)}V_3(\vartheta_{\tau_h \wedge t},\tau_h \wedge t)) - V_3(\vartheta_0,0) \\
				&= \mathrm{E}\int_0^{\tau_h \wedge t }e^{\delta s}[\delta V_3(\vartheta_s,s)+\mathfrak{L}V_3(\vartheta_s,s)]ds,\notag
			\end{split}
		\end{eqnarray}
		where $\delta>0$. Note that $\check{V} _3(\vartheta)\le V_3(\vartheta(t),t)$, where $\check{V} _3(\vartheta)=\vartheta^T(\mathcal{L}\otimes P)\vartheta$. Substituting the inequality $\check{V} _3(\vartheta)\le V_3(\vartheta(t),t)$ and \eqref{v1-5-linear} into the above equation yields
		\begin{eqnarray}
			\begin{split}
				&\mathrm{E}( e^{\delta(\tau_h \wedge t)}\check{V}_3(\vartheta_{\tau_h \wedge t} )-C_3  \\
				&\le (\delta-\varepsilon) \mathrm{E}\int_0^{\tau_h \wedge t }e^{\delta s} V_3(\vartheta_s,s)ds,\notag
			\end{split}
		\end{eqnarray}
		where $C_3=\check{V} _3(\vartheta_0)=\vartheta^T(0)(\mathcal{L}\otimes P)\vartheta(0)$. Let $\delta=\varepsilon$, then $\mathrm{E}( e^{\delta(\tau_h \wedge t)}\check{V}_3(\vartheta_{\tau_h \wedge t}) )\le C_3$. Thus, let $h\to \infty$, then we can obtain by using the Fatou lemma that
		\begin{eqnarray}\label{S1-3-linear}
			\begin{split}
				\mathrm{E}( e^{\delta t}\check{V}_3(\vartheta(t)) )\le C_3.
			\end{split}
		\end{eqnarray}
		That is, $\mathrm{E}( \check{V}_3(\vartheta(t)) )\le C_3e^{-\delta t}$, which implies
		\begin{eqnarray}\label{S1-4-linear}
			\begin{split}
				\lim_{t\to \infty}\sup \frac{1}{t}\log \mathrm{E}( \check{V}_3(\vartheta(t)) )\le -\delta.
			\end{split}
		\end{eqnarray}
		This together with the definition of $\check{V} _3$ also produces
		\begin{equation}
			\lim_{t\to \infty}\sup \frac{1}{t}\log ( \mathrm{E}|\vartheta(t)|^2)\le -\delta,
		\end{equation}
		That is, the m.s. exponential stability of the closed-loop system \eqref{3-3-linear} follows.
		According to the definition of $\vartheta(t)$, the m.s. consensus can be solved.
		
		Step 3: According to the functional $\mathrm{It\hat{o}} $ formula, it can be seen that
		\begin{eqnarray}
			\begin{split}
				&e^{\delta t} V_3(\vartheta(t),t)-V_3(\vartheta_0,0)\\
				&=\int_0^t e^{\delta s}[\delta V_3(\vartheta(s), s)+\mathfrak{L}V_3(\vartheta_s,s)]ds+M_{v3}(t),\notag
			\end{split}
		\end{eqnarray}
		where $M_{v3}(t)$ is a local martingale with $M_{v3}(0)=0$. Then, similar to \eqref{S1-3-linear}, it can be deduced that
		\begin{eqnarray}
			\begin{split}
				e^{\delta t} \check{V}_3(\vartheta(t))\le C_3+M_{v3}(t),\notag
			\end{split}
		\end{eqnarray}
		Using the non-negative semi-martingale convergence theorem \cite{liptser1989theory}, ones have $\lim_{t \to \infty}\sup [e^{\delta t} \check{V}_3(\vartheta(t))]<\infty$ a.s. Hence,
		\begin{eqnarray}\label{S1-5-linear}
			\begin{split}
				\lim_{t\to \infty}\sup \frac{1}{t}\log ( \check{V}_3(\vartheta(t)) )\le -\delta, a.s.
			\end{split}
		\end{eqnarray}
		Then, \eqref{S1-5-linear} together with the definition of $\check{V} _3$ yields
		\begin{eqnarray}
			\lim_{t\to \infty}\sup \frac{1}{t}\log (|\vartheta(t)| )\le -\frac{\delta}{2} \ a.s.\notag
		\end{eqnarray}
		That is, the a.s. exponential stability of the closed-loop system \eqref{3-3-linear} follows.
		According to the definition of $\vartheta(t)$, the a.s. consensus can be solved.
		
		Furthermore, the boundedness of input can be further elucidated below.
		\begin{eqnarray}
			\begin{split}
				|u(t)|=&|(\int_0^te^{\gamma s}\vartheta^T(s)R_1\vartheta(s)ds+C_1)R_2\vartheta(t)| \\
				\le &\lambda_{\max}(R_1R_2)\int_0^te^{\gamma s}|\vartheta(s)|^2ds|\vartheta(t)|+C_1R_2|\vartheta(t)|,\notag
			\end{split}
		\end{eqnarray}
		where $R_1=\mathcal{L}^2 \otimes PB  B  ^TP$, $C_1=\vartheta^T(0)R_1\vartheta(0)$ and $R_2=\mathcal{L} \otimes B  ^TP$.
		Note that $\lim_{t\to \infty}\sup \frac{1}{t}\log (|\vartheta(t)| )\le -\frac{\delta}{2} \ a.s.$ Hence, for any $0<\kappa<\frac{\delta}{2}$, one can find a positive random variable $\varpi=\varpi(\kappa)$ such that $|\vartheta(t)|\le \varpi e^{-(\frac{\delta}{2}-\kappa)t}$ for all $t\ge 0$.
		Thus, it can be obtained that
		\begin{eqnarray}
			\begin{split}
				|u(t)|\le &\lambda_{\max}(R_1R_2)\varpi^3\int_0^te^{(\gamma-\delta+2\kappa) s}ds\ e^{-(\frac{\delta}{2}-\kappa)t}\\
				&+C_1R_2\varpi e^{-(\frac{\delta}{2}-\kappa)t} \\
				\le &\frac{\lambda_{\max}(R_1R_2)\varpi^3}{\gamma-\delta+2\kappa}(e^{(\gamma-\frac{3\delta}{2}+3\kappa) t}-e^{-(\frac{\delta}{2}-\kappa)t})\\
				&+C_1R_2\varpi e^{-(\frac{\delta}{2}-\kappa)t}.\notag
			\end{split}
		\end{eqnarray}
		Note that $\delta=\varepsilon=\frac{1}{\lambda_{\max}(P)}$ and $ \gamma< \frac{3}{2\lambda_{\max}(P)}$. For sufficiently small $\kappa$, $\gamma-\frac{3\delta}{2}+3\kappa <0$. Therefore, it can be deduced that the input $u(t)$ is bounded.
	\end{proof}
	
	Moreover, according to Theorem \ref{the-3.1} in the preceding analysis, the adjustable parameter $\gamma$ is only required to satisfy $ 0<\gamma< \frac{3}{2\lambda_{\max}(P)}$ to guarantee the boundedness of input. The additional constraint $\gamma \ge \frac{1}{\lambda_{\max}(P)}$ in Theorem \ref{the3-1-linear} is introduced solely to derive the explicit exponential convergence rate of the stochastic MASs.
	
	\begin{remark}
		The inclusion of the exponential term $e^{\gamma t}$ introduces only a scalar operation per time step, whose computational cost is negligible relative to the distributed quadratic-form updates already inherent to the algorithm.
		Compared with the protocol presented in Theorem \ref{the-3.1}, the proposed scheme incorporates an adjustable parameter $\gamma$, thereby enhancing the flexibility of the control design. Moreover, it can be shown that an improved convergence rate is attainable through appropriate tuning of this parameter.
		
	\end{remark}
	
	
	\begin{remark}
		While the design of the fully distributed protocol for undirected graphs in Theorem \ref{the3-1-linear} is inspired by \cite{li2022fully}, our work differs from the results of \cite{li2022fully} in the following three aspects. First, we provide a rigorous proof of the existence and uniqueness of solutions for the path-dependent stochastic system by establishing verifiable sufficient conditions (Lemma~\ref{lemma-1}) within a stochastic Lyapunov framework, which serves as the foundation for the fully distributed protocol design. Second, we establish the uniform boundedness of the proposed protocol under the parameter constraint $\gamma < \frac{3}{2\lambda_{\max}(P)}$, ensuring that the control inputs remain bounded throughout system evolution, a critical practical property not explicitly analyzed in \cite{li2022fully}. Third, our core contribution lies in developing a fully distributed consensus framework for directed graphs (Theorems~\ref{theorem-1} and \ref{the4-1-linear}). By incorporating the auxiliary time-varying gain $\Sigma_i(t)$, the proposed protocol eliminates reliance on global network information and topological symmetry, thereby extending applicability to realistic networked systems with asymmetric information exchange.

	\end{remark}

	
%

	\section{Simulations}\label{sec4}
	In this section, two simulation examples are given to verify the effectiveness of the proposed theoretical results.
	
	Consider the stochastic MASs \eqref{2-1} composed of six agents, where $A =\begin{bmatrix}-0.5& 0.1 \\ 0& -20\end{bmatrix}$, $B  =\begin{bmatrix}0 \\1 \end{bmatrix}$, and $C=\begin{bmatrix}0&0 \\0 &6.5\end{bmatrix}$.
	According to the SARE \eqref{2-2}, we can obtain $P=\begin{bmatrix}1&0.0047  \\0.0047 &0.9046\end{bmatrix}$. Then, it can be deduced that $\mathcal{K} =\begin{bmatrix}-0.0047&-0.9046 \end{bmatrix}$, and $\Gamma=\begin{bmatrix}0&0.0042 \\0.0042 &0.8182\end{bmatrix}$.
	
	\subsection{The case of directed topology}
	The interaction topology is modeled as a directed graph $\tilde{\mathcal{G}}$, as shown in Figure \ref{directed.png}.
	\begin{figure}[htbp]
		\centering
		\hspace{2mm}
		\begin{minipage}[b]{4cm}
			\includegraphics[width=4cm]{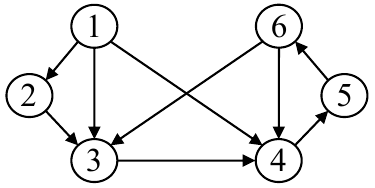}
			\vspace{-5.5mm}
			\caption{The directed graph\label{directed.png}}
		\end{minipage}
		\hspace{4mm}
		\begin{minipage}[b]{3.5cm}
			\includegraphics[width=3.5cm]{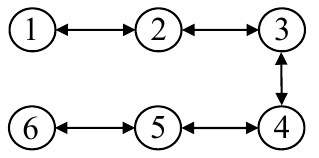}
			\caption{The undirected graph\label{undirected.png}}
		\end{minipage}
	\end{figure}
	The control input $u_i(t)$ adopts the adaptive control protocol \eqref{4-1}. The initial states for each agent are randomly chosen from $[-2,2]$. We choose $k_1=k_2=1$ and $c_i(0)=1$, $i=1, ..., 6$.
	The revolutions of adaptive gains $c_i(t)$, $i=1, ..., 6$ are shown in Figure \ref{Fig7.png}, implying that $c_i(t)$ will converge to a finite positive constant.
	The relative states $x_{i1}(t)-x_{11}(t)$, $x_{i2}(t)-x_{12}(t)$, $i=1, ..., 6$ of one sample path and the behaviors of the m.s. relative states ${E\|x_{i1}(t)-x_{11}(t)\|^2}$, ${E\|x_{i2}(t)-x_{12}(t)\|^2}$, $i=1, ..., 6$ of $10^2$ sample paths are shown in Figure \ref{Fig6.png}, indicating that a.s. and m.s. consensus can be achieved.
\begin{figure}[htbp]
	\centering
	\hspace{-1mm}
	\begin{minipage}[b]{4cm}
		\includegraphics[width=4.5cm]{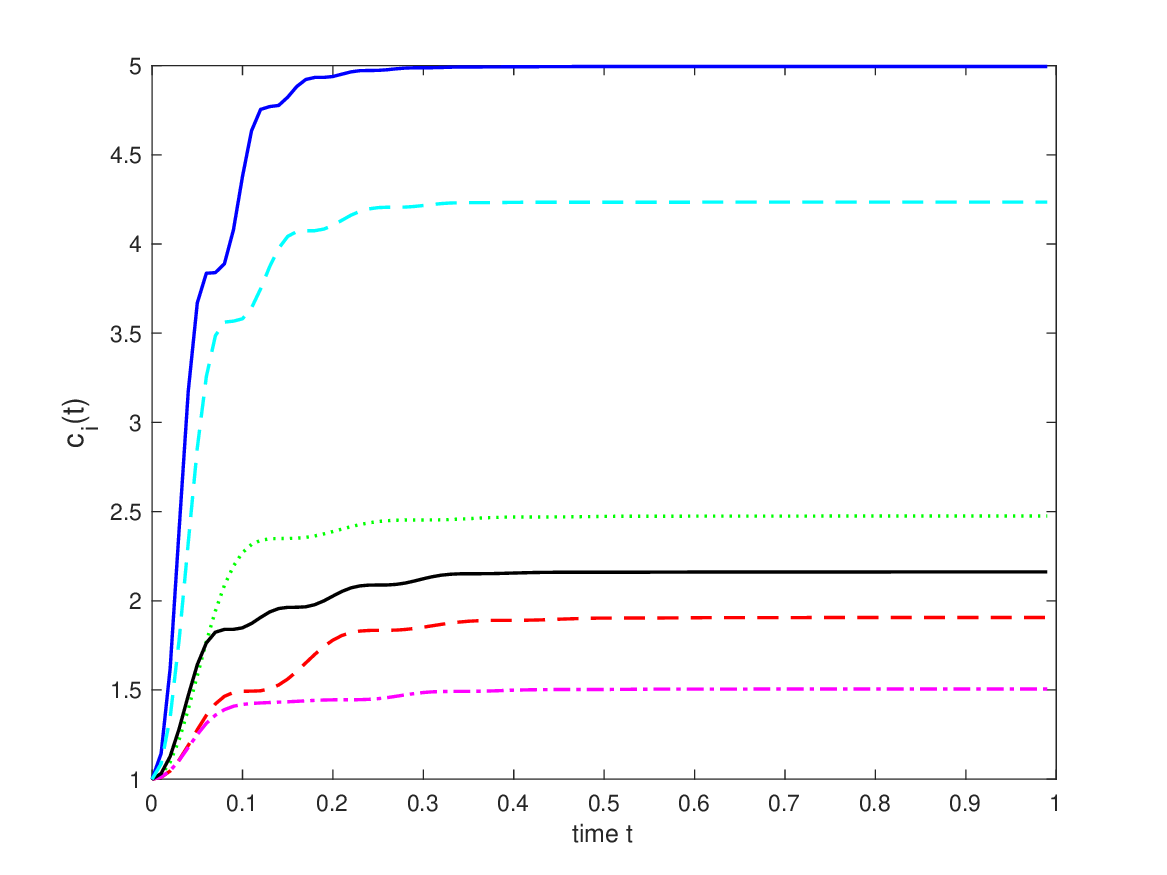}
		\vspace{-6.5mm}
		\caption{The adaptive gains for directed topology\label{Fig7.png}}
	\end{minipage}
	\hspace{1.5mm}
	\begin{minipage}[b]{4cm}
		\includegraphics[width=4.5cm]{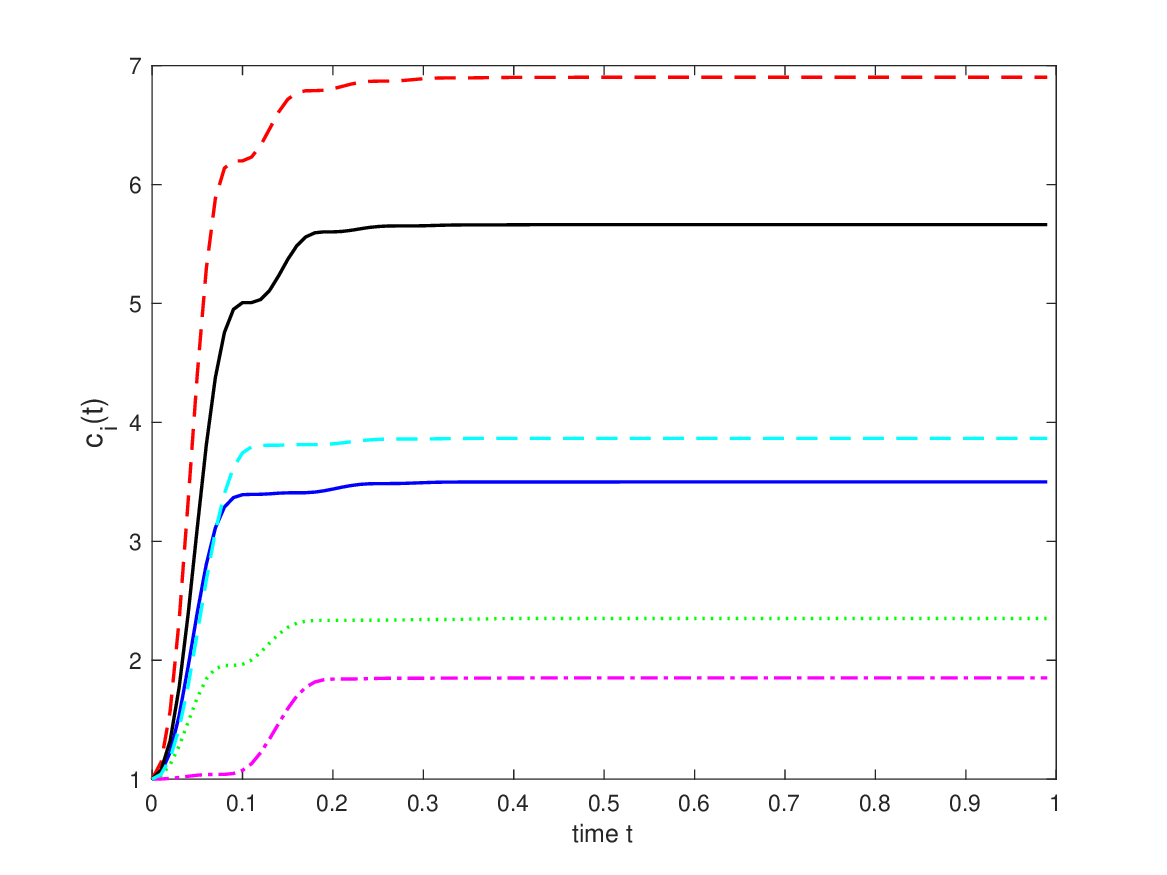}
		\vspace{-6.5mm}
		\caption{The adaptive gains for undirected topology\label{Fig3.png}}
	\end{minipage}
	\vspace{-5mm}
\end{figure}
	
	\subsection{The case of undirected topology}
	The interaction topology is modeled as an connected undirected graph $\bar{\mathcal{G}}$, as shown in Figure \ref{undirected.png}.
	The control input $u_i(t)$ adopts the adaptive control protocol \eqref{3-1}. The initial states for each agent are randomly chosen from $[-2,2]$.
	By Theorem \ref{the3-1-linear}, we can choose $\gamma=1$. Let $c_i(0)=1$, $i=1, ..., 6$.
	The revolutions of adaptive gains $c_i(t)$ are shown in Figure \ref{Fig3.png}.
	Considering the relative state of one sample path for each agent $x_{i1}(t)-x_{11}(t)$ and $x_{i2}(t)-x_{12}(t)$, we have Figure \ref{Fig2.png}, which indicates that states of the six agents tend to be a.s. consensus over time.
	For m.s. consensus analysis, we generate $10^2$ sample paths. Then, considering the behaviors of the m.s. relative states ${E\|x_{i1}(t)-x_{11}(t)\|^2}$ and ${E\|x_{i2}(t)-x_{12}(t)\|^2}$, we obtain Figure \ref{Fig2.png}, which demonstrates that the six agents reach m.s. consensus.
	\begin{figure}[htbp]
		\vspace{-4mm}
		\centering
		\begin{minipage}[b]{8cm}
				\hspace{2mm}\includegraphics[width=8cm]{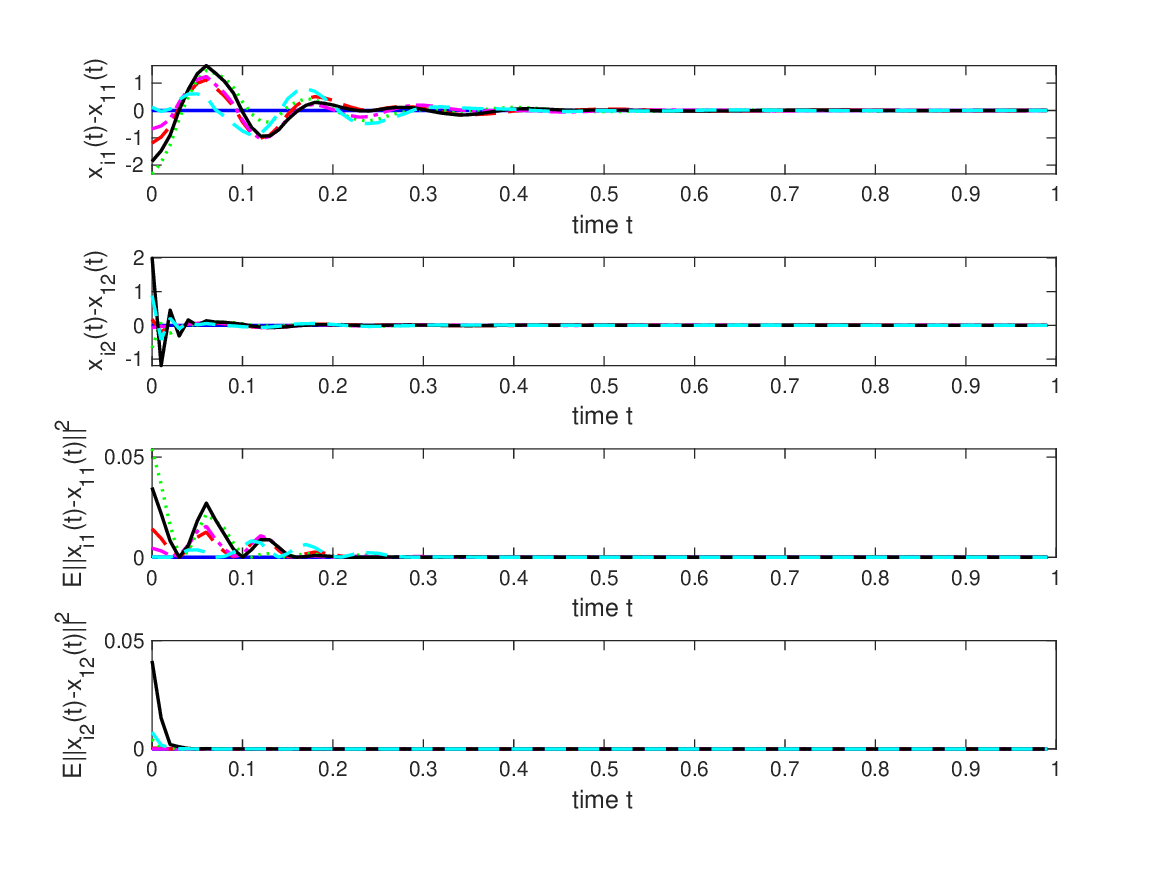}
			\vspace{-9mm}
			\caption{The relative states errors and m.s. relative state errors for directed topology\label{Fig6.png}}
		\end{minipage}
		\begin{minipage}[b]{8cm}
				\hspace{2mm}\includegraphics[width=8cm]{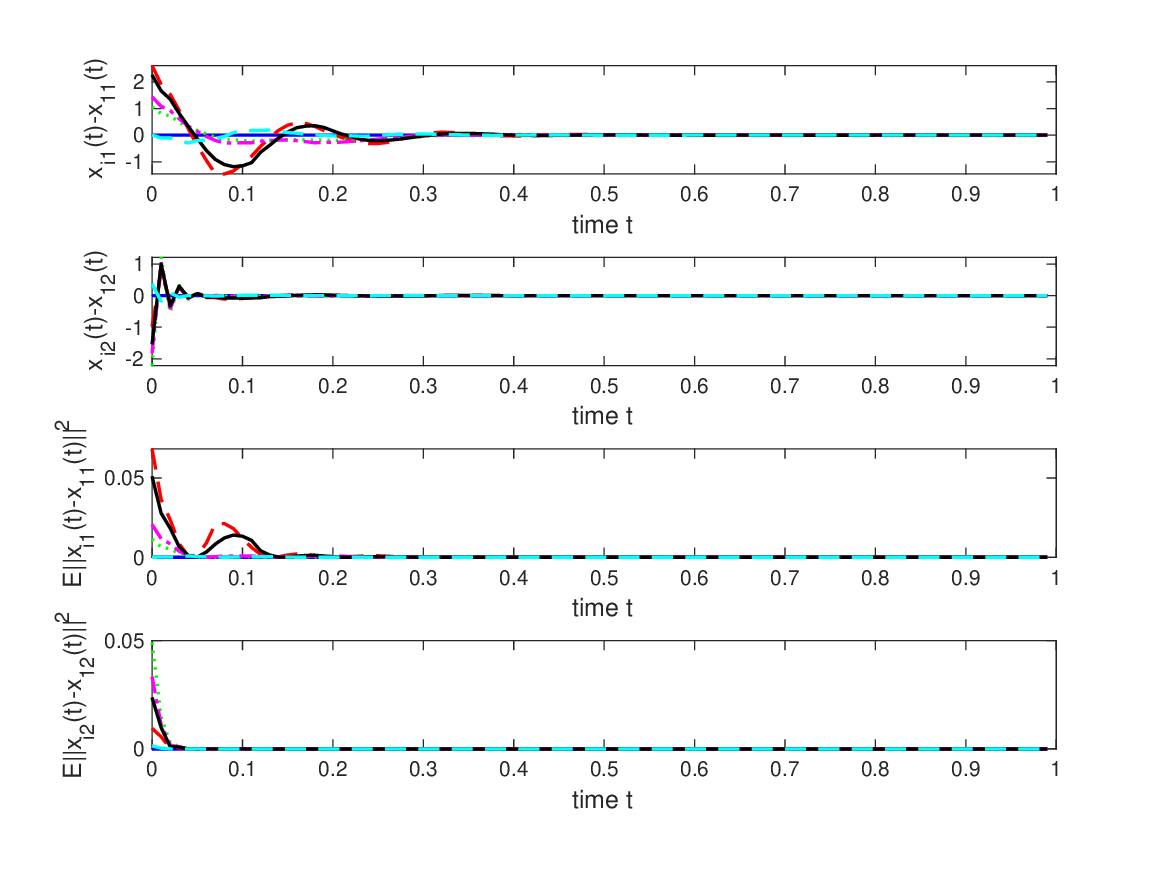}
			\vspace{-9mm}
			\caption{The relative states errors and m.s. relative state errors for undirected topology\label{Fig2.png}}
		\end{minipage}
		\vspace{-4mm}
	\end{figure}
	
	Additionally, for the cases of $\gamma = 0, 0.5,1$, the same initial value $x(0)=[-10\ -20\ -1.5\  -15\ -0.5\ -5\ 0.1\ 1\ 10\ 10\ 2\ 2]^T$ is used. Considering the behaviors of the m.s. relative states ${E\|x_{i1}(t)-x_{11}(t)\|^2}$ and ${E\|x_{i2}(t)-x_{12}(t)\|^2}$, we obtain Figures \ref{Fig-gamma-x1.png}-\ref{Fig-gamma-x2.png}, which demonstrates the effectiveness of the exponential gain term $e^{\gamma t}$ in accelerating convergence.
	\begin{figure}[htbp]
		\vspace{-2mm}
		\centering
		\begin{minipage}[b]{8cm}
				\hspace{2mm}\includegraphics[width=8cm]{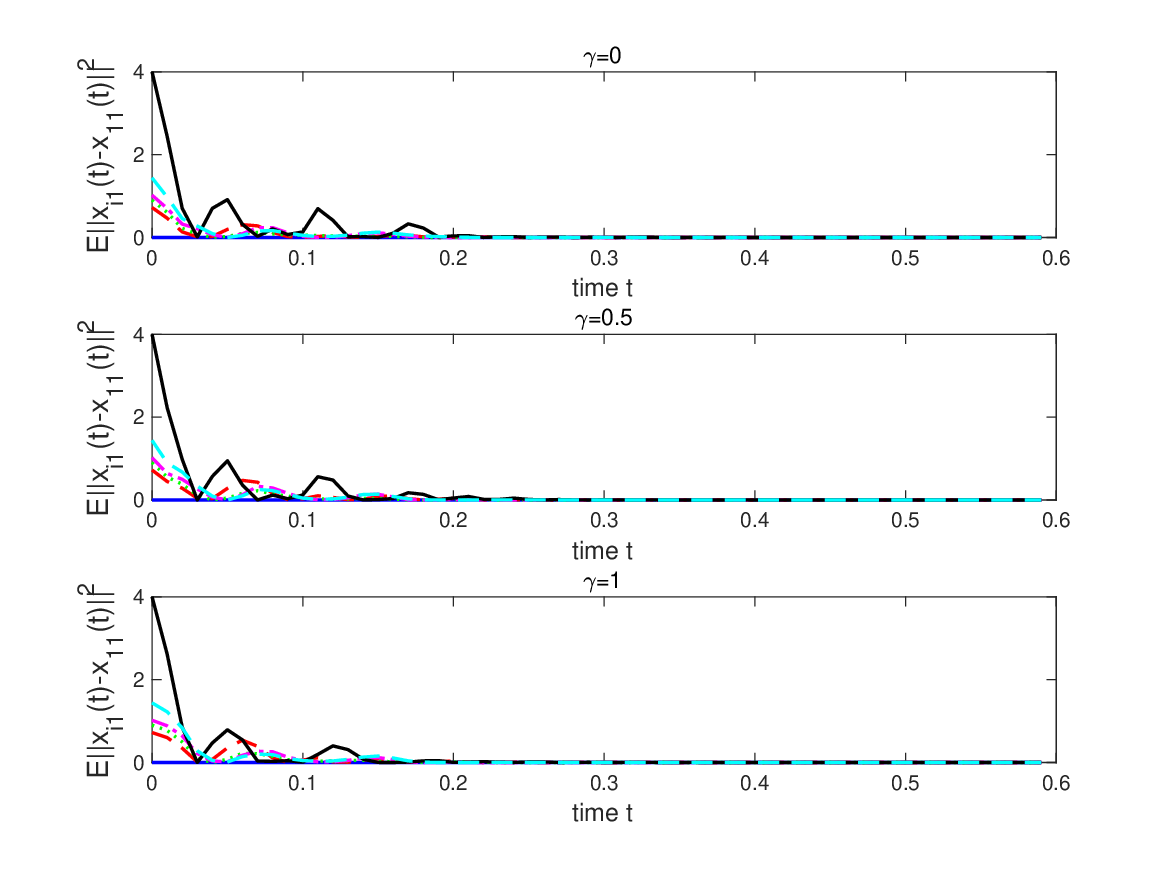}
			\vspace{-9mm}
			\caption{The m.s. relative state errors ${E\|x_{i1}(t)-x_{11}(t)\|^2}$ for the case of $\gamma = 0, 0.5,1$\label{Fig-gamma-x1.png}}
		\end{minipage}
		\begin{minipage}[b]{8cm}
				\hspace{2mm}\includegraphics[width=8cm]{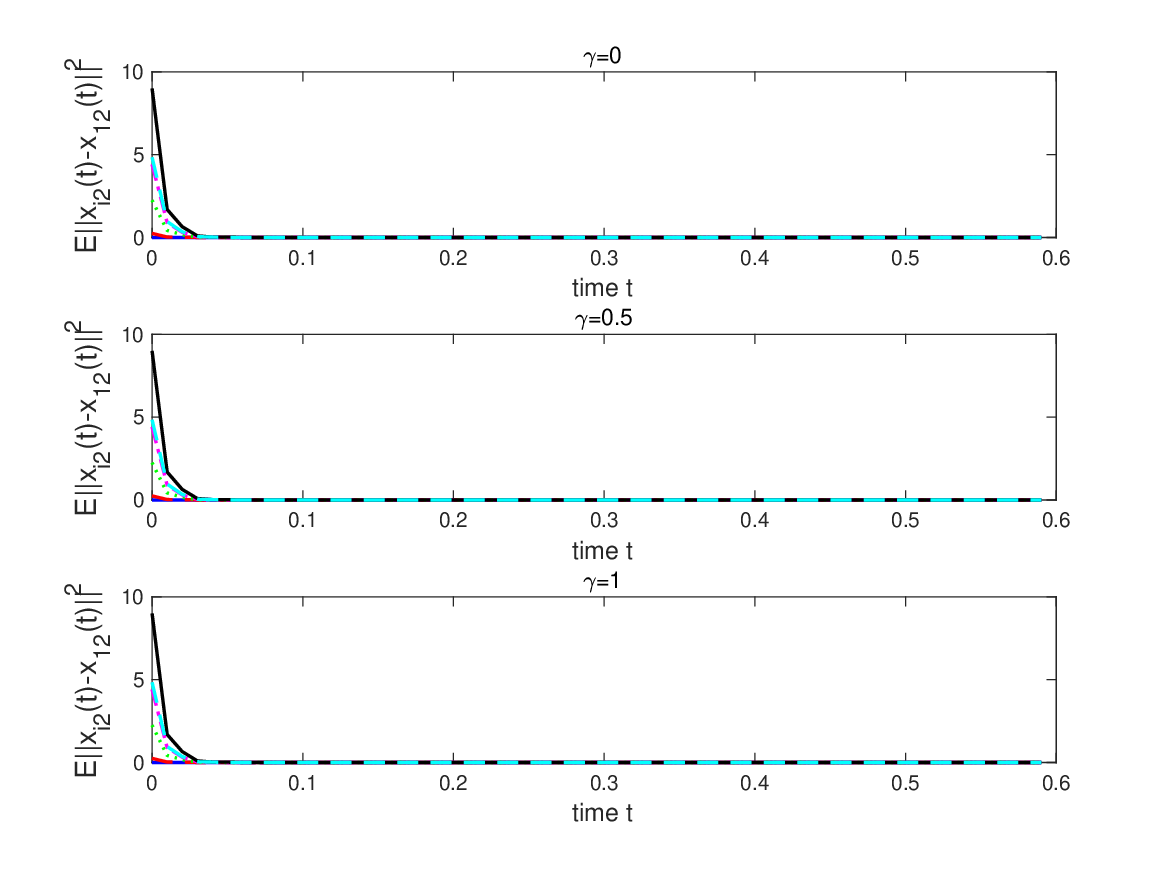}
			\vspace{-9mm}
			\caption{The m.s. relative state errors ${E\|x_{i2}(t)-x_{12}(t)\|^2}$ for the case of $\gamma = 0, 0.5,1$\label{Fig-gamma-x2.png}}
		\end{minipage}
	\end{figure}
	
	\section{Conclusion}\label{sec5}
	The design of a unified fully distributed consensus protocol for a class of stochastic MASs under directed and undirected graphs is investigated. The existence and uniqueness of solutions to path-dependent and highly nonlinear stochastic systems are first explored.
	For the case of directed graphs, a unified fully distributed control protocol is designed for the first time to solve m.s. and a.s. consensus for stochastic MASs.
	Then, for the case of undirected graphs, we develop an enhanced fully distributed protocol and derive explicit exponential convergence rates for both m.s. and a.s. consensus. 
	
	Future research directions include extending these results to MASs with compound disturbances and investigating more challenging scenarios involving heterogeneous nonlinear dynamics and time-varying communication topologies.
	Notably, the framework established in this paper can be naturally extended to nonlinear systems. While the current analysis focuses on linear dynamics with path-dependent adaptive feedback, the core analytical approach does not rely on the linear dynamics. For nonlinear dynamics satisfying appropriate growth conditions (such as, local Lipschitz condition and polynomial growth), a similar fully distributed protocol design and consensus analysis can be carried out for this class of stochastic nonlinear MASs. These extensions will be explored in future work.
	
	\bibliographystyle{IEEEtran}
	\bibliography{ref}\ 

	\vfill
	
\end{document}